\documentclass[12pt,onecolumn]{IEEEtran}

\usepackage{graphicx,amssymb,amsmath}
\usepackage{multicol}
\usepackage[noadjust]{cite}
\usepackage{setspace}
\usepackage{stfloats}
\usepackage{caption}
\usepackage{subfigure}
\usepackage{bbding}
\usepackage{amsthm}
\usepackage{amsmath}
\usepackage{flushend}
\usepackage{cases,subeqnarray}
\usepackage{bm,multirow,bigstrut}
\usepackage{textcomp}
\usepackage{latexsym,bm}
\usepackage{booktabs,changebar}
\usepackage{xcolor}
\newtheorem{corollary}{Corollary}

\usepackage{mathtools}
\usepackage{dsfont}
\usepackage{extarrows}
\usepackage{epstopdf}
\usepackage{mathrsfs}

\usepackage{cite}
\usepackage{bm}
\usepackage{cleveref}
\usepackage{multicol}       
\usepackage{multirow}       
\usepackage{array}          
\usepackage{colortbl}
\usepackage{makecell}
\definecolor{crimson}{RGB}{192,0,0}         
\definecolor{navy}{RGB}{47,85,151}         

\makeatletter                               
\newif\if@restonecol
\makeatother

\makeatletter
\newif\if@restonecol
\makeatother

\usepackage[linesnumbered,ruled,vlined]{algorithm2e}
\usepackage{algpseudocode}
\usepackage{amsmath}
\theoremstyle{plain}
\newtheorem{theorem}{Theorem}
\theoremstyle{plain}

\newtheorem{remark}{Remark}

\IEEEoverridecommandlockouts

\begin{document}

\title{Uplink Performance of Cell-Free Massive MIMO with Multi-Antenna Users Over Jointly-Correlated Rayleigh Fading Channels
\thanks{Z. Wang and J. Zhang are with the School of Electronic and Information Engineering, Beijing Jiaotong University, Beijing 100044, China. (e-mail: \{zhewang\_77, jiayizhang\}@bjtu.edu.cn).}
\thanks{B. Ai is with the State Key Laboratory of Rail Traffic Control and Safety, Beijing Jiaotong University, Beijing 100044, China (e-mail: aibo@ieee.org).}
\thanks{C.Yuen is with the Engineering Product Development Pillar, Singapore University of Technology and Design, Singapore 487372, Singapore (e-mail: yuenchau@sutd.edu.sg).}
\thanks{M. Debbah is with the Technology Innovation Institute, Abu Dhabi, United Arab Emirates, and also with CentraleSup{\'e}lec,
University Paris-Saclay, 91192 Gif-sur-Yvette, France (e-mail: merouane.debbah@tii.ae).}}
\author{Zhe Wang, Jiayi Zhang,~\IEEEmembership{Senior Member,~IEEE,} Bo Ai,~\IEEEmembership{Fellow,~IEEE} \\ Chau Yuen,~\IEEEmembership{Fellow,~IEEE}, and M{\'e}rouane Debbah,~\IEEEmembership{Fellow,~IEEE}}
\maketitle

\vspace{-1.5cm}
\begin{abstract}
In this paper, we investigate a cell-free massive MIMO system with both access points (APs) and user equipments (UEs) equipped with multiple antennas over jointly-correlated Rayleigh fading channels. We study four uplink implementations, from fully centralized processing to fully distributed processing, and derive their achievable spectral efficiency (SE) expressions with minimum mean-squared error successive interference cancellation (MMSE-SIC) detectors and arbitrary combining schemes. Furthermore, the global and local MMSE combining schemes are derived based on full and local channel state information (CSI) obtained under pilot contamination, which can maximize the achievable SE for the fully centralized and distributed implementation, respectively. We study a two-layer decoding implementation with an arbitrary combining scheme in the first layer and optimal large-scale fading decoding (LSFD) in the second layer. Besides, we compute novel closed-form SE expressions for the two-layer decoding implementation with maximum ratio (MR) combining. In the numerical results, we compare the SE performance for different implementation levels, combining schemes, and channel models. It is important to note that increasing the number of antennas per UE may degrade the SE performance.
\end{abstract}

\IEEEpeerreviewmaketitle

\section{Introduction}
Cell-Free massive multiple-input multiple-output (CF mMIMO) has been considered as a promising technology for future wireless communication \cite{9113273,7827017}. The basic idea of CF mMIMO is to deploy a large number of access points (APs), which are geographically distributed in the coverage area and connected to the central processing unit (CPU) via fronthaul connections. With the mutual cooperation and the assistance from the CPU, all APs coherently serve all user equipments (UEs) by spatial multiplexing on the same time-frequency resource \cite{7827017,7917284,8901451,[162]}. The main characteristic of CF mMIMO, compared with the traditional cellular mMIMO, is that the number of APs is envisioned to be much larger than the number of UEs and the operating regime with no cell boundaries \cite{[162]}.
The vast majority of papers on CF mMIMO rely upon a distributed implementation with the maximum ratio (MR) processing \cite{7827017,8097026}, while \cite{7917284} noticed that higher SE can be achieved with the assistance from partially or fully centralized processing at the CPU. Besides, the two-layer decoding scheme studied in \cite{8645336,7869024,8809413} for CF mMIMO is considered as an effective decoding technique with the arbitrary combining scheme in the first layer decoder and the large-scale fading decoding (LSFD) method in the second layer decoder. The authors of \cite{[162]} investigated four different CF mMIMO implementations from fully centralized to fully distributed with global or local minimum mean-square error (MMSE) combining and showed that the CF mMIMO system is competitive with the centralized implementation and MMSE processing compared with the cellular mMIMO and the small cell network.

The vast majority of existing papers focused on the analysis of CF mMIMO with single-antenna UEs, e.g. \cite{9113273,7827017,7917284,8901451,[162]}. However, in practice, contemporary UEs with moderate physical sizes (e.g. laptops, tablets, and vehicles) have already been equipped with multiple antennas to achieve higher multiplexing gain and boost the SE performance. So it is necessary to evaluate the performance for CF mMIMO systems with multi-antenna UEs. Besides, the effects of additional antennas at UEs should also be investigated to design the CF mMIMO system effectively. Recent works have investigated the performance of the CF mMIMO system with multi-antenna UEs, e.g.  \cite{113,8901451,194,8811486,9424703}. The authors in \cite{8901451} investigated the CF mMIMO system with a user-centric (UC) approach wherein both APs and UEs are equipped with multiple antennas. Besides, the downlink performance of the CF mMIMO system with multi-antenna UEs was analyzed in \cite{194}, where two downlink transmission protocols (with/without downlink pilot transmission) were considered. In addition, the authors of \cite{113} considered the uplink (UL) of the CF mMIMO system with multi-antenna UEs and derived the closed-form UL SE expressions. Moreover, the authors of \cite{8811486} and \cite{9424703} evaluated the performance for the CF mMIMO system with multi-antenna UEs and low-resolution DACs or ADCs. However, all these works are based on a classical distributed processing scheme as in \cite{7827017} (so-called “Level 2” in \cite{[162]}) but neglect pragmatic processing schemes in CF mMIMO, e.g. ``fully centralized processing" and ``LSFD scheme".

Moreover, all works about CF mMIMO with multi-antenna UEs made the overly idealistic and simplifying assumption of independent and identically distributed (i.i.d.) Rayleigh fading channels. It has been proved that the spatial correlation that exists in any practical channel may have a significant impact on performance \cite{8809413,9276421,8869794}. To the best of the authors' knowledge, there is no research on CF mMIMO with multi-antenna UEs over practical correlated channels so far. On the other hand, works on the traditional cellular mMIMO with multi-antenna UEs mainly focused on the classic Kronecker channel model \cite{1021913,7500452,7797479,7946138}, modeling the spatial correlation properties at the AP-side and UE-side separately. However, the authors of \cite{ozcelik2003deficiencies} mentioned that the classic Kronecker model neglects the joint correlation feature of the channel. As a remedy, a more practical channel model called the jointly-correlated Weichselberger model was proposed in \cite{1576533}, which not only considers the correlation features at both the AP-side and UE-side but models the joint correlation dependence between each AP-UE pair.

Motivated by the above observations, we investigate a CF mMIMO system with multi-antenna UEs over Rayleigh fading channels described by the jointly-correlated Weichselberger model. The comparisons of relevant literature with this paper are summarized in Table~\ref{Paper_comparison}. The major contributions of this paper are listed as follows.
\begin{itemize}
\item We consider a CF mMIMO system with multi-antenna UEs and the pilot contamination over the jointly-correlated Rayleigh fading channel, which is firstly considered in CF mMIMO, and investigate four different UL implementations from fully centralized to fully distributed inspired by \cite{[162]}.
\item We derive achievable UL SE expressions for four implementations with MMSE-SIC detectors and arbitrary combining schemes and compute novel closed-form SE expressions for Level 3 and Level 2 with MR combining. Note that our derived expressions hold for arbitrary combining scheme, arbitrary UL precoding scheme, and any Rayleigh fading channel models, e.g the Weichselberger model, the Kronecker model, and the i.i.d. Rayleigh fading model, etc.
\item We investigate global MMSE and local MMSE combining schemes for the scenario with multi-antenna UEs. Besides, we show that global MMSE combining in Level 4 and L-MMSE combining in Level 1 is also optimal to achieve the maximum SE value, respectively. Moreover, in Level 3, the optimal LSFD scheme with multi-antenna UEs is proposed, which not only minimizes the MSE between the two-layer decoding signal and the original signal but maximizes the achievable SE for Level 3.
\end{itemize}

\begin{table*}[tp]
  \centering
  \fontsize{9}{10}\selectfont
  \caption{Comparison of relevant literature with this paper.}
  \label{Paper_comparison}
    \begin{tabular}{ !{\vrule width1.2pt}  m{1.3 cm}<{\centering} !{\vrule width1.2pt}  m{2.cm}<{\centering} !{\vrule width1.2pt}  m{1.5 cm}<{\centering} !{\vrule width1.2pt} m{1.5cm}<{\centering} !{\vrule width1.2pt} m{1.5cm}<{\centering} !{\vrule width1.2pt} m{1 cm}<{\centering} !{\vrule width1.2pt} m{1.5cm}<{\centering} !{\vrule width1.2pt} m{1.5cm}<{\centering} !{\vrule width1.2pt} m{2cm}<{\centering} !{\vrule width1.2pt} }

    \Xhline{1.2pt}
        \rowcolor{gray!30} \bf Ref.  &  \bf Multi-antenna UEs &  \bf Correlation &  \bf Joint correlation  &  \bf Fully centralized processing &  \bf LSFD &  \bf MMSE processing & \bf Pilot contamination & \bf Closed-form  \cr
    \Xhline{1.2pt}

        \cite{[162]} & \makecell[c]{\XSolidBrush} & \makecell[c]{\Checkmark} & \makecell[c]{\XSolidBrush} & \makecell[c]{\Checkmark} & \makecell[c]{\Checkmark} & \makecell[c]{\Checkmark} & \makecell[c]{\Checkmark} & \makecell[c]{\XSolidBrush} \cr\hline
        \cite{8901451} & \makecell[c]{\Checkmark} & \makecell[c]{\XSolidBrush} & \makecell[c]{\XSolidBrush} & \makecell[c]{\XSolidBrush} & \makecell[c]{\XSolidBrush} & \makecell[c]{\Checkmark} & \makecell[c]{\Checkmark} & \makecell[c]{\XSolidBrush}\cr\hline
        \cite{194} & \makecell[c]{\Checkmark} & \makecell[c]{\XSolidBrush} & \makecell[c]{\XSolidBrush} & \makecell[c]{\XSolidBrush} & \makecell[c]{\XSolidBrush} & \makecell[c]{\Checkmark} & \makecell[c]{\XSolidBrush} & \makecell[c]{\Checkmark}\cr\hline
        \cite{9424703} & \makecell[c]{\Checkmark} & \makecell[c]{\XSolidBrush} & \makecell[c]{\XSolidBrush} & \makecell[c]{\XSolidBrush} & \makecell[c]{\XSolidBrush} & \makecell[c]{\XSolidBrush} & \makecell[c]{\XSolidBrush} & \makecell[c]{\Checkmark}\cr\hline
        \bf Proposed &  \makecell[c]{\Checkmark} & \makecell[c]{\Checkmark} & \makecell[c]{\Checkmark} & \makecell[c]{\Checkmark} & \makecell[c]{\Checkmark} & \makecell[c]{\Checkmark} & \makecell[c]{\Checkmark} & \makecell[c]{\Checkmark}\cr\hline
    \Xhline{1.2pt}
    \end{tabular}
  \vspace{0cm}
\end{table*}

The rest of this paper is organized as follows. In Section \ref{se:model}, we describe the Weichselberger channel model for a CF mMIMO system, investigate some practical channel models and discuss their relationships to the Weichselberger model and describe the phases of a CF mMIMO system, including the channel estimation and uplink data transmission. Next, Section \ref{se:SE Analysis} introduces four levels of AP cooperation. The achievable SE, MMSE or LMMSE combining, the optimal LSFD scheme and novel closed-form SE expressions for Level 3 and Level 2 with MR combining are also provided in this section. Then, numerical results and performance analysis are provided in Section \ref{se:Numerical Results}. Finally, the major conclusions and future directions are drawn in Section \ref{se:conclusion}.

\textbf{\emph{Notation}}: Column vectors and matrices are denoted by boldface lowercase letters $\mathbf{x}$ and boldface uppercase letters $\mathbf{X}$, respectively. $\left( \cdot \right) ^*$, $\left( \cdot \right) ^T$, and $\left( \cdot \right) ^H$ represent conjugate, transpose, and conjugate transpose, respectively. We use $\mathrm{diag}\left( \mathbf{A}_1,\cdots ,\mathbf{A}_n \right)$ to denote a block-diagonal matrix with the square matrices $\mathbf{A}_1,\cdots ,\mathbf{A}_n $ on the diagonal. $\mathbb{E} \left\{ \cdot \right\}$, $\mathrm{tr}\left\{ \cdot \right\}$ and $\triangleq$ represent the expectation operator, the trace operator and the definitions, respectively. $\left| \cdot \right|$ denotes the determinant of a matrix or the absolute value of a number. The $n\times n$ identity matrix is $\mathbf{I}_{n\times n}$. A column vector formed by the stack of the columns of $\mathbf{A}$ is represented by $\mathrm{vec}\left( \mathbf{A} \right)$. $\otimes$ and $\odot$ denote the Kronecker products and the element-wise products, respectively. $\left\| \cdot \right\|$ and $\left\| \cdot \right\| _{\mathrm{F}}$ are the Euclidean norm and the Frobenious norm, respectively. The $M\times N$ matrix with unit entries is denoted by $\mathbf{1}_{M\times N}$. $\mathbf{x}\sim \mathcal{N} _{\mathbb{C}}\left( 0,\mathbf{R} \right)$ represents a circularly symmetric complex Gaussian distribution vector with correlation matrix $\mathbf{R}$.

\section{System Model}\label{se:model}
As illustrated in Fig. \ref{System_Model}, we consider a CF mMIMO system consisting of $M$ APs and $K$ UEs arbitrarily distributed in a wide coverage area. $L$ and $N$ denote the number of antennas per AP and UE, respectively. Let $\mathbf{H}_{mk}\in \mathbb{C}^{L\times N}$ denote the channel response between AP $m$ and UE $k$. We assume that $\mathbf{H}_{mk}\in \mathbb{C}^{L\times N}$ for every AP $m$-UE $k$ pair is independent, $m=1, \ldots ,M, k=1,\ldots,K$. We investigate a standard block fading model, where the channel responses are assumed to be constant and frequency flat in a coherence block of $\tau _c$-length (in channel uses) and $\tau _c$ is equal to the product of the coherence bandwidth and coherence time \cite{8187178}. Besides, $\tau _p$ and $\tau _u=\tau _c-\tau _p$ channel uses are reserved for the training and the data transmission, respectively.
\begin{figure}[ht]
\centering
\includegraphics[scale=0.7]{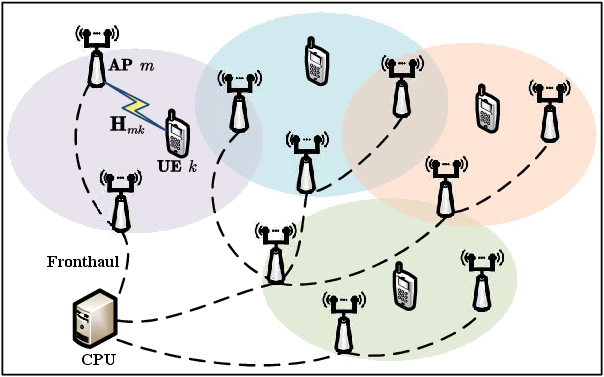}
\caption{Illustration of a CF mMIMO system.
\label{System_Model}}
\end{figure}
\subsection{Foundation of Weichselberger Rayleigh Fading Channels}
In this paper, we consider the jointly-correlated (also known as the Weichselberger model \cite{1576533}) Rayleigh fading channels as
\begin{equation}\label{eq:channel_model}
\mathbf{H}_{mk}=\mathbf{U}_{mk,\mathrm{r}}\left( \mathbf{\tilde{W}}_{mk}\odot \mathbf{H}_{mk,\mathrm{iid}} \right) \mathbf{U}_{mk,\mathrm{t}}^{H}
\end{equation}
where $\mathbf{H}_{mk,\mathrm{iid}}\in \mathbb{C}^{L\times N}$ is composed of i.i.d. $\mathcal{N}_{\mathbb{C}}\left( 0,1 \right)$ random entries, $\mathbf{U}_{mk,\mathrm{r}}=[ \mathbf{u}_{mk,\mathrm{r},1},\cdots ,\mathbf{u}_{mk,\mathrm{r},L} ] \in \mathbb{C}^{L\times L}$ and $\mathbf{U}_{mk,\mathrm{t}}=\left[ \mathbf{u}_{mk,\mathrm{t},1},\cdots ,\mathbf{u}_{mk,\mathrm{t},N} \right] \in \mathbb{C}^{N\times N}$ denote the eigenvector matrices of the one-sided correlation matrices $\mathbf{R}_{mk,\mathrm{r}}\triangleq \mathbb{E}[ \mathbf{H}_{mk}\mathbf{H}_{mk}^{H} ]\in \mathbb{C}^{L\times L}$ and $\mathbf{R}_{mk,\mathrm{t}}\triangleq \mathbb{E}[ \mathbf{H}_{mk}^{T}\mathbf{H}_{mk}^{*}] \in \mathbb{C}^{N\times N}$, respectively. Moreover, $\mathbf{\tilde{W}}_{mk}$ is the element-wise square root of the ``eigenmode coupling matrix" $\mathbf{W}_{mk}\triangleq \mathbf{\tilde{W}}_{mk}\odot \mathbf{\tilde{W}}_{mk}\in \mathbb{R}^{L\times N}$ with the $( l,n )$-th element $[ \mathbf{W}_{mk}] _{ln}$ specifying the average amount of power coupling from $\mathbf{u}_{mk,\mathrm{r},l}$ to $\mathbf{u}_{mk,\mathrm{t},n}$. Besides, $\mathbf{W}_{mk}$ links the joint correlation properties between AP $m$ and UE $k$ and shows the spatial arrangement of scattering objects. Let $\bm{\lambda }_{mk,\mathrm{r}}\in \mathbb{R}^L$ and $\bm{\lambda }_{mk, \mathrm{t}}\in \mathbb{R}^N$ denote the eigenvalues-vectors of the one-sided correlation matrices, whose $l$-th element $[ \bm{\lambda }_{mk,\mathrm{r}} ] _l$ and $n$-th element $[ \bm{\lambda }_{mk,\mathrm{t}} ] _n$ are given through
\begin{equation}\label{eq:lambda}
\begin{aligned}
\left[ \bm{\lambda }_{mk,\mathrm{r}} \right] _l=\sum_{n=1}^N{\left[ \mathbf{W}_{mk} \right] _{ln}},\left[ \bm{\lambda }_{mk,\mathrm{t}} \right] _n=\sum_{l=1}^L{\left[ \mathbf{W}_{mk} \right] _{ln}}.
\end{aligned}
\end{equation}

Note that the one-sided correlation matrices can be constructed as $\mathbf{R}_{mk,\mathrm{r}}=\mathbf{U}_{mk,\mathrm{r}}\mathrm{diag}( \bm{\lambda }_{mk,\mathrm{r}} ) \mathbf{U}_{mk,\mathrm{r}}^{H}$ and $\mathbf{R}_{mk,\mathrm{t}}=\mathbf{U}_{mk,\mathrm{t}}\mathrm{diag}( \bm{\lambda }_{mk,\mathrm{t}} ) \mathbf{U}_{mk,\mathrm{t}}^{H}$. In addition, $\mathbf{H}_{mk}$ can be formed as $\mathbf{H}_{mk}=[ \mathbf{h}_{mk,1},\cdots ,\mathbf{h}_{mk,N} ]$, where $\mathbf{h}_{mk,n}\in \mathbb{C}^L$ is the channel between AP $m$ and $n$-th antenna of UE $k$. As shown in \cite{5340650}, the channel can be modeled as $\mathbf{h}_{mk}=\mathrm{vec}( \mathbf{H}_{mk} ) \sim \mathcal{N}_{\mathbb{C}}( 0,\mathbf{R}_{mk} )$, where $\mathbf{R}_{mk}\triangleq \mathbb{E} [ \mathrm{vec}( \mathbf{H}_{mk} ) \mathrm{vec}( \mathbf{H}_{mk} ) ^H ] \in \mathbb{C}^{LN\times LN}$ is the full correlation matrix as
\begin{equation}
\begin{aligned}
\mathbf{R}_{mk}&=\left( \mathbf{U}_{mk,\mathrm{t}}^{*}\otimes \mathbf{U}_{mk,\mathrm{r}} \right) \mathrm{diag}\left( \mathrm{vec}\left( \mathbf{W}_{mk} \right) \right) \left( \mathbf{U}_{mk,\mathrm{t}}^{*}\otimes \mathbf{U}_{mk,\mathrm{r}} \right) ^H\\
&=\sum_{l=1}^L{\sum_{n=1}^N{\left[ \mathbf{W}_{mk} \right] _{ln}\left( \mathbf{u}_{mk,\mathrm{t},n}^{*}\otimes \mathbf{u}_{mk,\mathrm{r},l} \right) \left( \mathbf{u}_{mk,\mathrm{t},n}^{*}\otimes \mathbf{u}_{mk,\mathrm{r},l} \right) ^H}}. \label{eq:Correlation_Matrix}
\end{aligned}
\end{equation}

We can also structure the full correlation matrix in the block form as \cite{9148706} where the $(n,i)$-th submatrix is
 $\mathbf{R}_{mk}^{ni}=\mathbb{E} \{ \mathbf{h}_{mk,n}\mathbf{h}_{mk,i}^{H} \}$ with $\mathbf{h}_{mk,n}$ and $\mathbf{h}_{mk,i}$ being the $n$-th column and $i$-th column of $\mathbf{H}_{mk}$, respectively.
Note that the large-scale fading coefficient between AP $m$ and UE $k$ can be extracted from the full correlation matrix $\mathbf{R}_{mk}$ as
\begin{equation}\label{eq:LargeScale}
\beta _{mk}=\frac{1}{LN}\mathrm{tr}\left( \mathbf{R}_{mk} \right) =\frac{1}{LN}\left\| \mathbf{W}_{mk} \right\| _1.
\end{equation}

\begin{remark}
The full correlation matrix in \eqref{eq:Correlation_Matrix} shows a Kronecker structure on eigenmode level and is determined by the unitary matrix $\mathbf{U}_{mk,\mathrm{r}}$ at the AP side, $\mathbf{U}_{mk,\mathrm{t}}$ at the UE side and the coupling matrix $\mathbf{W}_{mk}$. The eigenvalues-vectors $\bm{\lambda }_{mk,\mathrm{r}}$ and $\bm{\lambda }_{mk, \mathrm{t}}$ don't influence the correlation matrix directly but are implicitly given through \eqref{eq:lambda}. Moreover, the large-scale fading coefficient is reflected in the coupling matrix $\mathbf{W}_{mk}$ so that the power constraint in \eqref{eq:LargeScale} should be satisfied in all scenarios investigated.
\end{remark}

\subsection{Special Cases of Weichselberger Fading Channels}\label{Special_channel_cases}
In this subsection, we investigate some practical channel models and discuss their relationships to the Weichselberger model. Note that the jointly-correlated model investigated in \eqref{eq:channel_model} embraces most channels of great interest. And structures of $\mathbf{W}_{mk}$ display the spatial arrangement of scattering objects in the real environment. In Fig. \ref{W_Scatters}, we show different structures of $\mathbf{W}_{mk}$ and their respective physical radio environments, where the black dots denote the scattering objects. Note that $\left[ \mathbf{W}_{mk} \right] _{ln}$ is the average power between the $l$-th eigenmode of the AP-side and the $n$-th eigenmode of the UE-side so a nonzero entry of $\mathbf{W}_{mk}$ shows a link between the respective eigenmodes (gray squares in $\mathbf{W}$) and a zero-entry denotes an inactive link between the respective eigenmodes (white squares in $\mathbf{W}$). The coupling matrices and corresponding physical radio environments of the Weichselberger model are represented by $\mathbf{W}_1$.
\begin{figure*}[ht]
\centering
\includegraphics[scale=0.7]{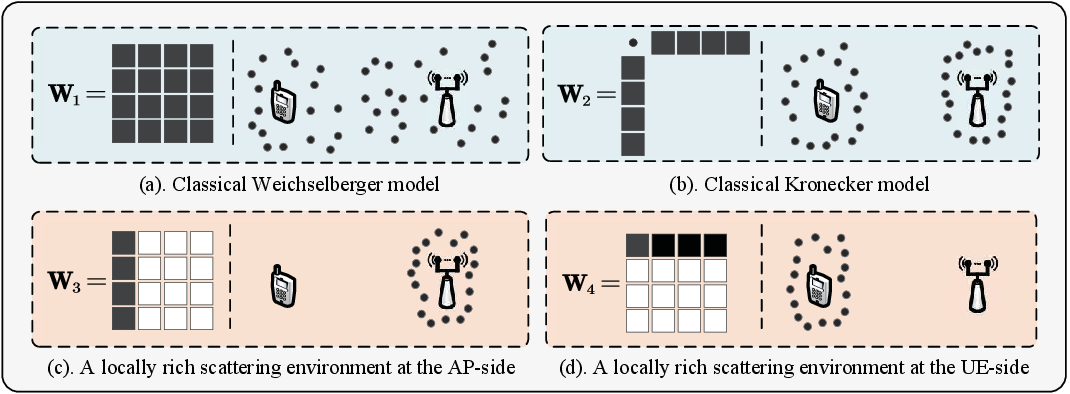}
\caption{Different structures of the coupling matrix and respective physical radio environments.
\label{W_Scatters}}
\end{figure*}

\subsubsection{The Kronecker Channel Model}
Firstly, we investigate the relationship between the Weichselberger model and the classical Kronecker model. As observed in $\mathbf{W}_{2}$, both the AP-side and UE-side go through locally rich scattering environments so we can model spatial correlation properties separately. And $\mathbf{W}_{mk}$ has rank-one form as
\begin{equation}
\mathbf{W}_{mk}=\frac{1}{LN\beta _{mk}}\bm{\lambda }_{mk,\mathrm{r}}\bm{\lambda }_{mk,\mathrm{t}}^{T},
\end{equation}
where $\bm{\lambda }_{mk,\mathrm{r}}$ and $\bm{\lambda }_{mk,\mathrm{t}}$ are defined in \eqref{eq:lambda}. The classic Kronecker model neglects the joint spatial correlation structure and the channel in \eqref{eq:channel_model} can be written as
\begin{equation}\label{eq:Kronecker}
\mathbf{H}_{mk}=\frac{1}{\sqrt{LN\beta _{mk}}}\sqrt{\mathbf{R}_{mk,\mathrm{r}}}\mathbf{H}_{mk,\mathrm{iid}}\sqrt{\mathbf{R}_{mk,\mathrm{t}}},
\end{equation}
The full correlation matrix $\mathbf{R}_{mk}$ can be reformulated as
$
\mathbf{R}_{mk}=\frac{1}{LN\beta _{mk}}\left( \mathbf{R}_{mk,\mathrm{t}}^{T}\otimes \mathbf{R}_{mk,\mathrm{r}} \right)
$ \cite{5997288}.

Moreover, a special case for \eqref{eq:Kronecker} is ``$\mathbf{R}_{mk,\mathrm{r}}=\beta _{mk}\mathbf{I}_L$, $\mathbf{R}_{mk,\mathrm{t}}\ne \beta _{mk}\mathbf{I}_N$", which denotes the scenario with only correlation properties at AP-sides and no correlation properties at UE-sides. And a similar case with ``$\mathbf{R}_{mk,\mathrm{r}}\ne\beta _{mk}\mathbf{I}_L$, $\mathbf{R}_{mk,\mathrm{t}}=\beta _{mk}\mathbf{I}_N$" means that the scenario with only correlation properties at UE-sides and no correlation properties at AP-sides.

\subsubsection{Practical Radio Environments}\label{Practical_Channel_Models}
Then, we elaborate on some practical physical radio environments, which are reflected in the structures of $\mathbf{W}$. As shown in $\mathbf{W}_{3}$, only a single column is available in $\mathbf{W}$, which shows a locally rich scattering environment at the AP-side. A similar structure for the case of ``a locally rich scattering environment at the UE-side" can be represented by $\mathbf{W}_{4}$ with only a single row available.

\subsubsection{The I.I.D. Rayleigh Fading Channel}
Last but not least, we focus on the i.i.d. Rayleigh fading channel, which is widely assumed in prior works \cite{113,8901451,194,8811486,9424703}. The channel model in \eqref{eq:channel_model} reduces to the uncorrelated Rayleigh fading channel if all the elements in $\mathbf{W}_{mk}$ are equal shown as
\begin{equation}\label{IID_Rayleigh}
\mathbf{W}_{mk}=\beta _{mk}\mathbf{1}_{L\times N}.
\end{equation}
Then, \eqref{eq:channel_model} will reduce to
$
\mathbf{H}_{mk}=\sqrt{\beta _{mk}}\mathbf{H}_{mk,\mathrm{iid}}.
$
Even though the i.i.d. Rayleigh fading channel is widely investigated in prior works, \eqref{eq:channel_model} is undoubtedly a more practical choice for CF mMIMO systems with multi-antenna APs and UEs, which can characterize the joint-correlation features in practical scenarios.
\vspace{-0.5cm}
\subsection{Uplink Transmission}
\subsubsection{Channel Estimation}
In this phase, $\tau _p$-length channel uses are adopted for the channel estimation. Note that $\tau _p$ mutually orthogonal $\tau _p$-length pilot sequences are used for this phase and each pilot matrix is composed of $N$ pilot sequences selected from the pilot book.
Let $\mathbf{\Phi }_k\in \mathbb{C}^{\tau _p\times N}$ denote the pilot matrix of UE $k$ with $\boldsymbol{\Phi }_{k}^{H}\boldsymbol{\Phi }_l=\tau _p\mathbf{I}_N$, if $\ l=k$ and ${\bf{0}}$ otherwise.
As in \cite{113} and \cite{194}, all UEs can be assigned to mutually orthogonal pilot matrices if $\tau _p\geqslant KN$. However, we consider a more practical scenario where more than one UE is assigned to a same pilot matrix. We define $\mathcal{P}_k$ as the index subset of UEs that use the same pilot matrix as UE $k$ including itself.

All UEs send pilot matrices to APs and the received signal $\mathbf{Y}_{m}^{p}\in \mathbb{C}^{L\times \tau _p}$ at AP $m$ is
\begin{equation}\label{eq:Pilot_Signal}
\mathbf{Y}_{m}^{p}=\sum_{k=1}^K{\mathbf{H}_{mk}\mathbf{\Omega }_k\mathbf{\Phi }_{k}^{T}+\mathbf{N}_{m}^{p}},
\end{equation}
where $\mathbf{\Omega }_k\in \mathbb{C} ^{N\times N}$ is the precoding matrix of UE $k$ during the phase of pilot transmission\footnote{We assume $\mathbf{\Omega }_k$ and $\mathbf{P}_k$ in the following are available at all APs and the CPU.}, $\mathbf{N}_{m}^{p}\in \mathbb{C}^{L\times \tau _p}$ is the additive noise at AP $m$ with independent $\mathcal{N}_{\mathbb{C}}( 0,\sigma ^2 )$ elements and $\sigma ^2$ is the noise power, respectively. The pilot transmission is under the pilot power constraint as $\mathrm{tr(}\mathbf{\Omega }_k\mathbf{\Omega }_{k}^{H})\leqslant \hat{p}_k$, where $\hat{p}_k$ is the maximum pilot transmit power of UE $k$.

To obtain sufficient statistics for channel estimation, AP $m$ computes the projection of $\mathbf{Y}_{m}^{p}$ onto $\boldsymbol{\Phi }_{k}^{*}$ as
\begin{equation}\label{eq:Estimate_Signal}
\begin{aligned}
\mathbf{Y}_{mk}^{p}&=\mathbf{Y}_{m}^{p}\boldsymbol{\Phi }_{k}^{*}=\sum_{l=1}^K{\mathbf{H}_{ml}\mathbf{\Omega }_l\left( \boldsymbol{\Phi }_{l}^{T}\boldsymbol{\Phi }_{k}^{*} \right) +}\mathbf{N}_{m}^{p}\boldsymbol{\Phi }_{k}^{*}\\
&=\sum_{l\in \mathcal{P}_k}{\tau _p\mathbf{H}_{ml}\mathbf{\Omega }_l}+\mathbf{Q}_{m}^{p},
\end{aligned}
\end{equation}
where $\mathbf{Q}_{m}^{p}\triangleq \mathbf{N}_{m}^{p}\boldsymbol{\Phi }_{k}^{*}$. By implementing the vectorization operation for $\mathbf{Y}_{mk}^{p}$, we derive $\mathbf{y}_{mk}^{p}=\mathrm{vec}\left( \mathbf{Y}_{mk}^{p} \right) \in \mathbb{C}^{LN}$
\begin{equation}\label{eq:Estimate_Signal_vec}
\begin{aligned}
\mathbf{y}_{mk}^{p}&=\sum_{l\in \mathcal{P}_k}{\tau _p\left( \mathbf{\Omega }_{l}^{T}\otimes \mathbf{I}_L \right) \mathrm{vec}\left( \mathbf{H}_{ml} \right) +}\mathrm{vec}\left( \mathbf{Q}_{m}^{p} \right)\\
&=\sum_{l\in \mathcal{P}_k}{\tau _p\mathbf{\tilde{\Omega}}_{l}\mathbf{h}_{ml}+\mathbf{q}_{m}^{p}},
\end{aligned}
\end{equation}
where $\mathbf{\tilde{\Omega}}_l\triangleq \mathbf{\Omega }_l^{T}\otimes \mathbf{I}_L$ and $\mathbf{q}_{m}^{p}\triangleq \mathrm{vec}\left( \mathbf{Q}_{m}^{p} \right)$. As in \cite{5340650}, the MMSE estimate of $\mathbf{h}_{mk}$ is
\begin{equation}\label{eq:MMSE_Estimation}
\mathbf{\hat{h}}_{mk}=\mathrm{vec}( \mathbf{\hat{H}}_{mk} ) =\mathbf{R}_{mk}\mathbf{\tilde{\Omega}}_{k}^{H}\mathbf{\Psi }_{mk}^{-1}\mathbf{y}_{mk}^{p},
\end{equation}
where $\mathbf{\Psi }_{mk}=\mathbb{E}\{ \mathbf{y}_{mk}^{p}( \mathbf{y}_{mk}^{p} ) ^H \}/{\tau _p}=\sum\nolimits_{l\in \mathcal{P}_k}{\tau _p\mathbf{\tilde{\Omega}}_{l}\mathbf{R}_{ml}}\mathbf{\tilde{\Omega}}_{l}^{H}+\sigma ^2\mathbf{I}_{LN}$. The estimate $\mathbf{\hat{h}}_{mk}$ and estimation error $\mathbf{\tilde{h}}_{mk}=\mathrm{vec}( \mathbf{\tilde{H}}_{mk} )=\mathbf{h}_{mk}-\mathbf{\hat{h}}_{mk}$ are independent random vectors distributed as
$\mathbf{\hat{h}}_{mk}\sim \mathcal{N}_{\mathbb{C}}( \mathbf{0},\mathbf{\hat{R}}_{mk} )$, $\mathbf{\tilde{h}}_{mk}\sim \mathcal{N}_{\mathbb{C}}( \mathbf{0},\mathbf{C}_{mk} )$
with $\mathbf{\hat{R}}_{mk}\triangleq \tau _p\mathbf{R}_{mk}\mathbf{\tilde{\Omega}}_{k}^{H}\mathbf{\Psi }_{mk}^{-1}\mathbf{\tilde{\Omega}}_k\mathbf{R}_{mk}$ and
$
\mathbf{C}_{mk}=\mathbb{E}\{ \mathrm{vec}( \mathbf{\tilde{H}}_{mk} ) \mathrm{vec}( \mathbf{\tilde{H}}_{mk}) ^H \} =\mathbf{R}_{mk}- \tau _p\mathbf{R}_{mk}\mathbf{\tilde{\Omega}}_{k}^{H}\mathbf{\Psi }_{mk}^{-1}\mathbf{\tilde{\Omega}}_k\mathbf{R}_{mk}.
$
We can also structure the full covariance matrix of the channel estimate into a block form
with the $(n,i)$-th submatrix being $\mathbf{\hat{R}}_{mk}^{ni}=\mathbb{E}\{ \mathbf{\hat{h}}_{mk,n}\mathbf{\hat{h}}_{mk,i}^{H} \} $, 
where $\mathbf{\hat{h}}_{mk,n}$ and $\mathbf{\hat{h}}_{mk,i}$ are the $n$-th and $i$-th column of $\mathbf{\hat{H}}_{mk}$, respectively.

\subsubsection{Data Transmission}
The transmitted signal $\mathbf{s}_k=\left[ s_{k,1},\cdots ,s_{k,N} \right] ^T\in \mathbb{C}^{N}$ from UE $k$ can be constructed as $\mathbf{s}_k=\mathbf{P}_k\mathbf{x}_k$,
where $\mathbf{x}_k\sim \mathcal{N}_{\mathbb{C}}\left( 0,\mathbf{I}_N \right) $ is the data symbol transmitted from UE $k$ and $\mathbf{P}_k\in \mathbb{C} ^{N\times N}$ is the precoding matrix during the phase of data transmission which should satisfy the power constraint of UE $k$ as $\mathrm{tr}( \mathbf{P}_{k}\mathbf{P}_{k}^{H} ) \leqslant p_k$ with $p_k$ is the maximum transmitted power of UE $k$, respectively.

The received signal $\mathbf{y}_m\in \mathbb{C}^L$ at AP $m$ is
\begin{align}
\mathbf{y}_m=\sum_{k=1}^K{\mathbf{H}_{mk}\mathbf{s}_k}+\mathbf{n}_m=\sum_{k=1}^K{\mathbf{H}_{mk}\mathbf{P}_k\mathbf{x}_k}+\mathbf{n}_m,
\end{align}
where $\mathbf{n}_m\sim \mathcal{N}_{\mathbb{C}}( 0,\sigma ^2\mathbf{I}_L )$ is the independent receiver noise.
\section{Four Signal Processing Schemes}\label{se:SE Analysis}
Due to the versatile network topology of CF mMIMO, four signal processing schemes were investigated in \cite{[162]} with the multi-antenna APs/single-antenna UEs scenario, depending on the required performance and complexity tradeoff. To the best of the authors' knowledge, recent works for UL analysis of CF mMIMO with multi-antenna UEs \cite{113,8901451,8811486}, consider the classical processing structure of CF mMIMO (so-called ``Level 2" in \cite{[162]}). So it is undoubtedly vital to investigate other processing schemes with multi-antenna UEs to provide important results and insights for practical CF mMIMO systems. To better distinguish the differences between four implementations clearly, the comparison of four levels investigated below is shown as Fig. \ref{Four_Levels} (we follow the naming principle in \cite{[162]}).

\begin{figure}[t]
\centering
\includegraphics[scale=0.7]{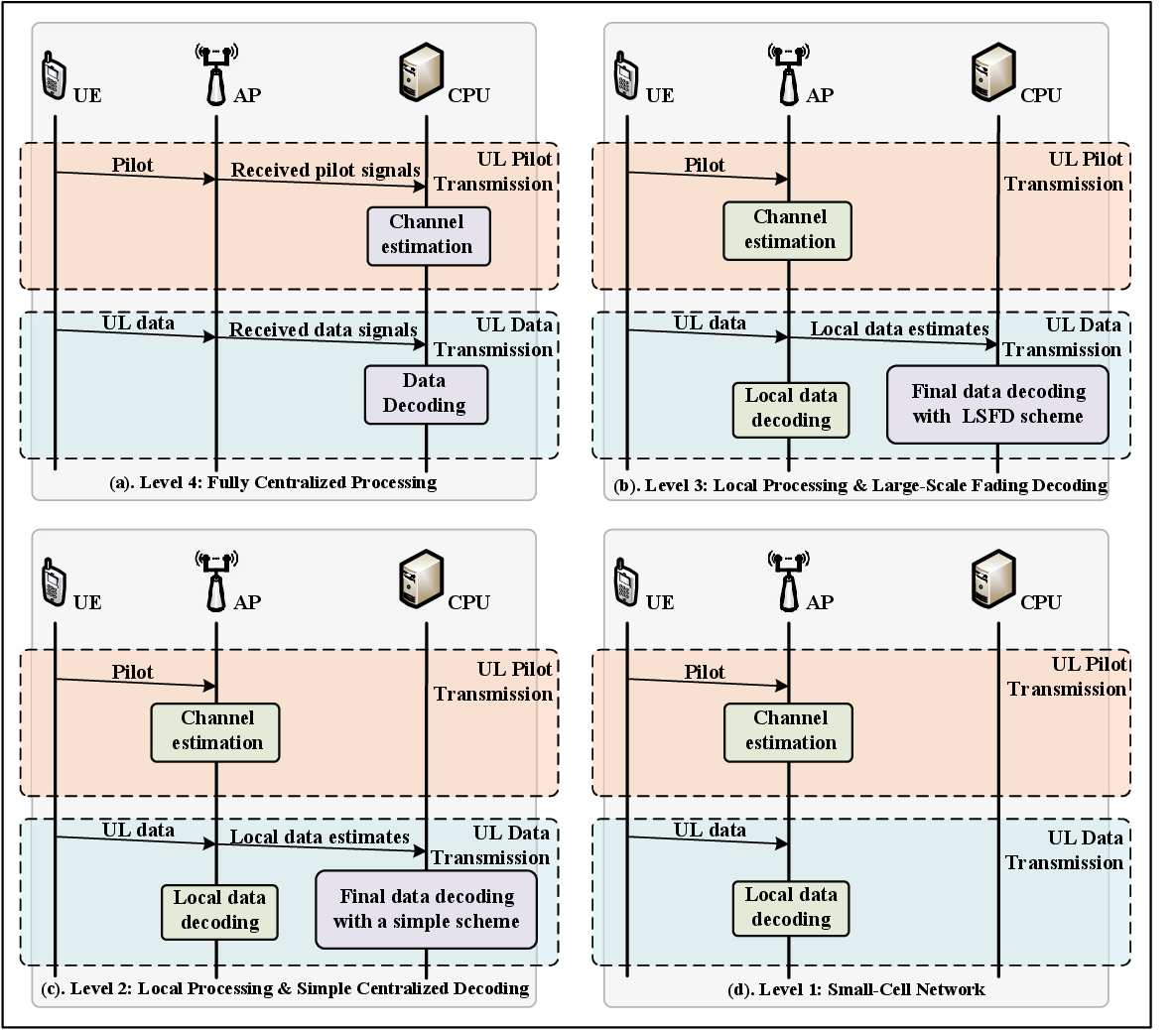}
\caption{The comparison of four UL implementations.
\label{Four_Levels}}
\vspace{-0.3cm}
\end{figure}
\subsection{Level 4: Fully Centralized Processing}
In this processing scheme, all the channel estimation and data detection are processed in the CPU.
All APs serve only as relays by sending all received pilot signals and received data signals to the CPU. The pilot signal $\mathbf{Y}_{\mathrm{c}}^{p}\in \mathbb{C}^{ML\times \tau _p}$ at the CPU can be formed as
\begin{equation}
\mathbf{Y}_{\mathrm{c}}^{p}=\left[ \begin{array}{c}
	\mathbf{Y}_{1}^{p}\\
	\vdots\\
	\mathbf{Y}_{M}^{p}\\
\end{array} \right] =\sum_{k=1}^K{\left[ \begin{array}{c}
	\mathbf{H}_{1k}\\
	\vdots\\
	\mathbf{H}_{Mk}\\
\end{array} \right] \mathbf{\Omega }_k\mathbf{\Phi }_{k}^{T}}+\left[ \begin{array}{c}
	\mathbf{N}_{1}^{p}\\
	\vdots\\
	\mathbf{N}_{M}^{p}\\
\end{array} \right].
 \end{equation}
For UE $k$, the collective channel $\mathbf{h}_k\in \mathbb{C}^{MLN}$ can be shown as
$
\mathbf{h}_k=[ \mathrm{vec}( \mathbf{H}_{1k} )^T ,\cdots ,\mathrm{vec}( \mathbf{H}_{Mk} )^T ] ^T\sim \mathcal{N} _{\mathbb{C}}( \mathbf{0},\mathbf{R}_k)
$
where $\mathbf{R}_k=\mathrm{diag}( \mathbf{R}_{1k},\cdots ,\mathbf{R}_{Mk}) \in \mathbb{C}^{MLN\times MLN}$ is the whole correlation matrix of UE $k$.

As the method of local estimation \eqref{eq:Estimate_Signal} and \eqref{eq:Estimate_Signal_vec}, the CPU can compute all the MMSE estimates. The collective channel estimate of UE $k$ can be constructed as
$
\mathbf{\hat{h}}_k\triangleq \left[ \mathbf{\hat{h}}_{1k}^T,\cdots ,\mathbf{\hat{h}}_{Mk}^T \right] ^T\sim \mathcal{N} _{\mathbb{C}}\left( \mathbf{0},\tau _p\mathbf{R}_k\mathbf{\bar{\Omega}}_{k}^{H}\mathbf{\Psi }_{k}^{-1}\mathbf{\bar{\Omega}}_k\mathbf{R}_k \right)
$
where $\mathbf{\Psi }_{k}^{-1}=\mathrm{diag}\left( \mathbf{\Psi }_{1k}^{-1},\cdots ,\mathbf{\Psi }_{Mk}^{-1} \right)$ and $\mathbf{\bar{\Omega}}_{k}=\mathrm{diag}( \underset{M}{\underbrace{\mathbf{\tilde{\Omega}}_{k},\cdots ,\mathbf{\tilde{\Omega}}_{k}}} )$.

As for the data detection, the received data signal at the CPU is
\begin{equation}
\underset{=\,\,\mathbf{y}}{\underbrace{[\mathbf{y}_{1}^{T},\cdots ,\mathbf{y}_{M}^{T}]^T}}=\sum_{k=1}^K{\underset{=\,\,\mathbf{H}_k}{\underbrace{[\mathbf{H}_{1k}^{T},\cdots ,\mathbf{H}_{Mk}^{T}]^T}}\mathbf{P}_k\mathbf{x}_k}+\underset{=\,\,\mathbf{n}}{\underbrace{[\mathbf{n}_{1}^{T},\cdots ,\mathbf{n}_{M}^{T}]^T}}
\end{equation}
%
or a compact form as
\begin{align}\label{eq:Level4_y}
\mathbf{y}&=\sum_{k=1}^K{\mathbf{H}_k\mathbf{P}_k\mathbf{x}_k}+\mathbf{n}.
\end{align}
The CPU selects an arbitrary receive combining matrix $\mathbf{V}_k\in \mathbb{C}^{LM\times N}$ based on the collective channel estimates for the detection of $\mathbf{x}_k$ as
\begin{equation}\label{eq:Level4_Sk}
\begin{aligned}
\mathbf{\check{x}}_k&=\mathbf{V}_{k}^{H}\mathbf{y}\\
&=\mathbf{V}_{k}^{H}\mathbf{\hat{H}}_k\mathbf{P}_k\mathbf{x}_k \!+\!\mathbf{V}_{k}^{H}\mathbf{\tilde{H}}_k\mathbf{P}_k\mathbf{x}_k\!+\!\sum_{l\ne k}^K{\mathbf{V}_{k}^{H}\mathbf{H}_l\mathbf{P}_l\mathbf{x}_l}\!+\!\mathbf{V}_{k}^{H}\mathbf{n}.
\end{aligned}
\end{equation}

Based on \eqref{eq:Level4_Sk}, we derive the achievable SE for Level 4 by using standard capacity lower bounds \cite{1624653} as the following corollary.
\begin{corollary}\label{SE_4}
If the MMSE estimator is used to compute channel estimates for all UEs, an achievable SE for UE $k$ in Level 4 with MMSE-SIC detectors is
\begin{equation}\label{eq:SE_4}
\mathrm{SE}_{k}^{\left( 4 \right)}=( 1-\frac{\tau _p}{\tau _c} ) \mathbb{E}\left\{ \log _2\left| \mathbf{I}_N+\mathbf{D}_{k,\left( 4 \right)}^{H}\mathbf{\Sigma }_{k,\left( 4 \right)}^{-1}\mathbf{D}_{k,\left( 4 \right)} \right| \right\} ,
\end{equation}
where $\mathbf{D}_{k,\left( 4 \right)}\triangleq \mathbf{V}_{k}^{H}\mathbf{\hat{H}}_k\mathbf{P}_{k}$ and
$
\mathbf{\Sigma }_{k,\left( 4 \right)}\triangleq \mathbf{V}_{k}^{H}\left( \sum_{l=1}^K{\mathbf{\hat{H}}_l\mathbf{\bar{P}}_l\mathbf{\hat{H}}_{l}^{H}}-\mathbf{\hat{H}}_k\mathbf{\bar{P}}_k\mathbf{\hat{H}}_{k}^{H}+\sum_{l=1}^K{\mathbf{C}_{l}^{\prime}}+\sigma ^2\mathbf{I}_{ML} \right) \mathbf{V}_k
$
with $\mathbf{\bar{P}}_k\triangleq \mathbf{P}_k\mathbf{P}_{k}^{H}$. Besides, $\mathbf{C}_{l}^{\prime}=\mathrm{diag}\left( \mathbf{C}_{1l}^{\prime},\cdots ,\mathbf{C}_{Ml}^{\prime} \right) \in \mathbb{C}^{ML\times ML}$ with $\left( j,q \right) $-th element of $\mathbf{C}_{ml}^{\prime}=\mathbb{E}\{ \mathbf{\tilde{H}}_{ml}\mathbf{\bar{P}}_l\mathbf{\tilde{H}}_{ml}^{H} \} $ being
$
[ \mathbf{C}_{ml}^{\prime}] _{jq} =\sum_{p_1=1}^N{\sum_{p_2=1}^N{[ \mathbf{\bar{P}}_l ] _{p_2p_1}[ \mathbf{C}_{ml}^{p_2p_1} ] _{jq}}}.
$
The expectations are with respect to all sources of randomness.
\end{corollary}
\begin{proof}
The proof is given in Appendix \ref{appendix_proof_coro_1}.
\end{proof}

Note that any receive combining matrix $\mathbf{V}_k$ can be utilized in \eqref{eq:SE_4} and the SE can be computed by the Monte-Carlo method.  The CPU can use all CSI to design $\mathbf{V}_k$ and we consider two combining schemes for Level 4: MR combining $\mathbf{V}_{k}=\mathbf{\hat{H}}_{k}$ and MMSE combining which minimizes the mean-squared error $\mathrm{MSE}_k=\mathbb{E}\{ \| \mathbf{x}_k-\mathbf{V}_{k}^{H}\mathbf{y} \| ^2| \mathbf{\hat{H}}_k \}$ as
\begin{equation}\label{eq:MMSE_Com_4}
\mathbf{V}_k=\left( \sum_{l=1}^K{\left( \mathbf{\hat{H}}_l\mathbf{\bar{P}}_l\mathbf{\hat{H}}_{l}^{H}+\mathbf{C}_{l}^{\prime} \right)}+\sigma ^2\mathbf{I}_{ML} \right) ^{-1}\mathbf{\hat{H}}_k\mathbf{P}_k,
\end{equation}
where $\mathbf{\hat{H}}_l=[ \mathbf{\hat{H}}_{1l}^T,\cdots ,\mathbf{\hat{H}}_{Ml}^T ]^T \in \mathbb{C}^{ML\times N}$.
\begin{proof}
The proof of \eqref{eq:MMSE_Com_4} can be easily derived following from the standard results of matrix derivation in \cite{hjorungnes2011complex} and is therefore omitted.
\end{proof}

In the scenario with single-antenna UEs, so-called MMSE combining can maximize the achievable SE \cite{[162]}. In the scenario with multi-antenna UEs, we prove that the MMSE combining matrix in \eqref{eq:MMSE_Com_4} is also the optimal combining matrix that maximizes \eqref{eq:SE_4}, which can be obtained as the following theorem.\footnote{Note that Theorem~\ref{Theorem1} shows the optimality of MMSE combining for maximizing $\mathrm{SE}_{k}^{\left( 4 \right)}$ and absolutely provides the guidance for the implementation of CF mMIMO with multi-antenna UEs.}
\begin{theorem}\label{Theorem1}
MMSE combining matrix as \eqref{eq:MMSE_Com_4}
is the optimal combining matrix leading to the maximum SE value given as
\begin{align}\label{eq:SE_max}
\mathrm{SE}_{k}^{\left( 4 \right)}=\left( 1-\frac{\tau _p}{\tau _c} \right) \mathbb{E} \left\{ \log _2\left| \mathbf{I}_N+\mathbf{P}_{k}^{H}\mathbf{\hat{H}}_{k}^{H}\left( \sum_{l=1,l\ne k}^K{\mathbf{\hat{H}}_l\mathbf{\bar{P}}_l\mathbf{\hat{H}}_{l}^{H}}+\sum_{l=1}^K{\mathbf{C}_{l}^{\prime}}+\sigma ^2\mathbf{I}_{ML} \right) ^{-1}\mathbf{\hat{H}}_k\mathbf{P}_k \right| \right\}.
\end{align}
\end{theorem}
\begin{proof}
The proof is given in Appendix \ref{appendix_Proof_Th_1}.
\end{proof}

\begin{figure*}[t!]
\setcounter{equation}{18}

\hrule
\end{figure*}

\subsection{Level 3: Local Processing \& Large-Scale Fading Decoding} \label{Level3_Section}

In this subsection, we investigate a two-layer decoding scheme. Let $\mathbf{V}_{mk}\in \mathbb{C}^{L\times N}$ denote the combining matrix designed by AP $m$ for UE $k$. Then, the local estimate of $\mathbf{s}_k$ at AP $m$ is
\begin{equation}\label{eq:S_kn_3}
\begin{split}
\mathbf{\tilde{x}}_{mk}&=\mathbf{V}_{mk}^{H}\mathbf{y}_m\\
&=\mathbf{V}_{mk}^{H}\mathbf{H}_{mk}\mathbf{P}_k\mathbf{x}_k+\sum_{l=1,l\ne k}^K{\mathbf{V}_{mk}^{H}\mathbf{H}_{ml}\mathbf{P}_l\mathbf{x}_l}+\mathbf{V}_{mk}^{H}\mathbf{n}_m.
\end{split}
\end{equation}
We notice that \eqref{eq:S_kn_3} holds for any combining vector and AP $m$ exploits its local estimate $\mathbf{\hat{H}}_{mk}$ to design $\mathbf{V}_{mk}$. One possible choice is MR combining $\mathbf{V}_{mk}=\mathbf{\hat{H}}_{mk}$. Besides, Local MMSE (L-MMSE) combining matrix that minimizes $\mathrm{MSE}_{mk}=\mathbb{E}\{ \| \mathbf{x}_k-\mathbf{V}_{mk}^{H}\mathbf{y}_m \| ^2| \mathbf{\hat{H}}_{mk} \}
$ is
\begin{equation}\label{eq:L_MMSE_3}
\begin{aligned}
\mathbf{V}_{mk}=\left( \sum_{l=1}^K{\left( \mathbf{\hat{H}}_{ml}\mathbf{\bar{P}}_l\mathbf{\hat{H}}_{ml}^{H}+\mathbf{C}_{ml}^{'} \right)}+\sigma ^2\mathbf{I}_L \right) ^{-1}\mathbf{\hat{H}}_{mk}\mathbf{P}_k.
\end{aligned}
\end{equation}

\begin{proof}
The proof of \eqref{eq:L_MMSE_3} is similar to the proof of \eqref{eq:MMSE_Com_4} and is therefore omitted.
\end{proof}

\begin{remark}
Note that MMSE combining matrix in Level 4 not only minimizes the MSE but maximizes the achievable SE. A similar combining scheme called L-MMSE combining is proposed based on the local CSI at each AP. Note that L-MMSE combing can maximize the achievable SE if the corresponding AP is the only one decoding the signal (this scenario will be investigated as a ``small-cell network"). So L-MMSE combing scheme is a heuristic combining scheme in this Level but undoubtedly makes sense.
\end{remark}

A second layer decoding structure called ``LSFD" is implemented \cite{[162],8809413}. The local estimates $\mathbf{\tilde{s}}_{mk}$ are sent to the CPU where they are weighted by the LSFD coefficient matrix as
\begin{equation}\label{eq:S_kn_LSFD}
\begin{aligned}
\mathbf{\hat{x}}_{k}=\sum_{m=1}^M{\mathbf{A}_{mk}^{H}\mathbf{\tilde{x}}_{mk}}=\sum_{m=1}^M{\mathbf{A}_{mk}^{H}\mathbf{V}_{mk}^{H}\mathbf{H}_{mk}\mathbf{P}_k\mathbf{x}_k}+\sum_{m=1}^M{\sum_{l=1,l\ne k}^K{\mathbf{A}_{mk}^{H}\mathbf{V}_{mk}^{H}\mathbf{H}_{ml}\mathbf{P}_l\mathbf{x}_l}+}\mathbf{n}_{k}^{\prime},
\end{aligned}
\end{equation}
where $\mathbf{A}_{mk}\in \mathbb{C}^{N\times N}$ is the complex LSFD coefficient matrix for AP $m$ and UE $k$ and $\mathbf{n}_{k}^{\prime}=\sum_{m=1}^M{\mathbf{A}_{mk}^{H}\mathbf{V}_{mk}^{H}\mathbf{n}_m}$.

Then we define $\mathbf{A}_k\triangleq [ \mathbf{A}_{1k}^T,\cdots ,\mathbf{A}_{Mk}^T ]^T \in \mathbb{C}^{MN\times N}$ and $\mathbf{G}_{kl}\triangleq \left[ \mathbf{V}_{1k}^{H}\mathbf{H}_{1l};\cdots ;\mathbf{V}_{Mk}^{H}\mathbf{H}_{Ml} \right] \in \mathbb{C}^{MN\times N}$, so \eqref{eq:S_kn_LSFD} can be rewritten as
\begin{equation}\label{eq:S_kn_LSFD_1}
\mathbf{\hat{x}}_k=\mathbf{A}_{k}^{H}\mathbf{G}_{kk}\mathbf{P}_k\mathbf{x}_k+\sum_{l=1,l\ne k}^K{\mathbf{A}_{k}^{H}\mathbf{G}_{kl}\mathbf{P}_l\mathbf{x}_l}+\mathbf{n}_{k}^{\prime}.
\end{equation}

Note that only channel statistics are available at the CPU so we apply the classical use-and-then-forget (UatF) bound \cite{8187178} to compute the achievable SE for the LSFD scheme as the following corollary.

\begin{corollary}
An achievable SE for UE $k$ in Level 3 with MMSE-SIC detectors is
\begin{equation}\label{eq:SE_3}
\mathrm{SE}_{k}^{\left( 3 \right)}=\left( 1-\frac{\tau _p}{\tau _c} \right) \log _2\left| \mathbf{I}_N+\mathbf{D}_{k,\left( 3 \right)}^{H}\mathbf{\Sigma }_{k,\left( 3 \right)}^{-1}\mathbf{D}_{k,\left( 3 \right)} \right|,
\end{equation}
where $\mathbf{D}_{k,( 3 )}\triangleq \mathbf{A}_{k}^{H}\mathbb{E}\{ \mathbf{G}_{kk}\} \mathbf{P}_{k}$ and $\mathbf{\Sigma }_{k,( 3)}\triangleq \sum_{l=1}^K\mathbf{A}_{k}^{H}\mathbb{E}\{ \mathbf{G}_{kl}\mathbf{\bar{P}}_l\mathbf{G}_{kl}^{H}\} \mathbf{A}_k-\mathbf{D}_{k,( 3 )}\mathbf{D}_{k,( 3 )}^{H}+\sigma ^2\mathbf{A}_{k}^{H}\mathbf{S}_k\mathbf{A}_k$
with $\mathbf{S}_k\triangleq \mathrm{diag}( \mathbb{E}\{ \mathbf{V}_{1k}^{H}\mathbf{V}_{1k} \} ,\cdots ,\mathbb{E}\{ \mathbf{V}_{Mk}^{H}\mathbf{V}_{Mk} \} ) \in \mathbb{C}^{MN\times MN}$.

\end{corollary}
\begin{proof}
The proof of \eqref{eq:SE_3} is similar to the proof of \eqref{eq:MMSE_Com_4} and is therefore omitted.
\end{proof}

The complex LSFD coefficient matrix $\mathbf{A}_k$ can be optimized by the CPU to maximize the achievable SE of UE $k$ for Level 3 in \eqref{eq:SE_3} as the following theorem.
\begin{theorem}\label{theorem2}
The achievable SE in \eqref{eq:SE_3} is maximized by
\begin{equation}\label{LSFD_op}
\mathbf{A}_k=\left( \sum_{l=1}^K{\mathbb{E}\left\{ \mathbf{G}_{kl}\mathbf{\bar{P}}_l\mathbf{G}_{kl}^{H} \right\}}+\sigma ^2\mathbf{S}_k \right) ^{-1}\mathbb{E}\left\{ \mathbf{G}_{kk} \right\} \mathbf{P}_k,
\end{equation}
which leads to the maximum value as
\begin{align}\label{SE_max_LSFD}
\mathrm{SE}_{k}^{\left( 3 \right)}=\left( 1-\frac{\tau _p}{\tau _c} \right) \log _2\left| \mathbf{I}_N+\mathbf{P}_k^{H}\mathbb{E} \left\{ \mathbf{G}_{kk}^{H} \right\} \left( \sum_{l=1}^K{\mathbb{E} \left\{ \mathbf{G}_{kl}\mathbf{\bar{P}}_l\mathbf{G}_{kl}^{H} \right\}}-\mathbb{E} \left\{ \mathbf{G}_{kk} \right\} \mathbf{\bar{P}}_k\mathbb{E} \left\{ \mathbf{G}_{kk}^{H} \right\} +\sigma ^2\mathbf{S}_k \right) ^{-1}\mathbb{E} \left\{ \mathbf{G}_{kk} \right\} \mathbf{P}_k \right|
\end{align}
\end{theorem}
\begin{proof}
The proof is similar to the proof of Theorem~\ref{Theorem1} and is therefore omitted.
\end{proof}

Note that the CPU is only aware of channel statistics, so the optimal LSFD coefficient matrix in \eqref{LSFD_op} can also minimize the conditional MSE of UE $k$ $\mathrm{MSE}_{k}^{\mathrm{LSFD}}=\mathbb{E} \left\{ \left\| \mathbf{x}_k-\mathbf{\hat{x}}_k \right\| ^2\left| \mathbf{\Theta } \right. \right\} $, which is proven in Appendix~\ref{appendix_minimization_LSFD} with $\mathbf{\Theta }$ being all the channel statistics. Furthermore, if MR combining $\mathbf{V}_{mk}=\mathbf{\hat{H}}_{mk}$ is adopted at each AP in the first layer, we can compute the expectations in \eqref{eq:SE_3} in closed-form and derive the closed-form SE expression as the following theorem.

\begin{theorem}\label{theorem3}
For MR combining $\mathbf{V}_{mk}=\mathbf{\hat{H}}_{mk}$, the achievable SE for UE $k$ in Level 3 with MMSE-SIC detectors can be computed in closed-form as
\setcounter{equation}{26}
\begin{equation}\label{eq:SE_3_closed_form}
\mathrm{SE}_{k}^{\left( 3 \right)}=\left( 1-\frac{\tau _p}{\tau _c} \right) \log _2\left| \mathbf{I}_N+\mathbf{D}_{k,\left( 3 \right)}^{H}\mathbf{\Sigma }_{k,\left( 3 \right)}^{-1}\mathbf{D}_{k,\left( 3 \right)} \right|,
\end{equation}
where $\mathbf{D}_{k,( 3 )}=\mathbf{A}_{k}^{H}\mathbf{Z}_k\mathbf{P}_{k}$ and $\mathbf{\Sigma }_{k,( 3 )}=\mathbf{A}_{k}^{H}( \sum_{l=1}^K\mathbf{T}_{kl,( 1 )}^{\mathrm{L}3}+\sum_{l\in \mathcal{P}_k}\mathbf{T}_{kl,( 2 )}^{\mathrm{L}3} ) \mathbf{A}_k-\mathbf{D}_{k,( 3 )}\mathbf{D}_{k,( 3 )}^{H}+\sigma ^2\mathbf{A}_{k}^{H}\mathbf{S}_k\mathbf{A}_k$ with $\mathbf{Z}_k=[ \mathbf{Z}_{1k}^T,\cdots ,\mathbf{Z}_{Mk}^T ]^T \in \mathbb{C}^{MN\times N}$ and $\mathbf{S}_k =\mathrm{diag}\left( \mathbf{Z}_{1k},\cdots ,\mathbf{Z}_{Mk} \right) \in \mathbb{C}^{MN\times MN}$ with the $\left( n,n^{\prime} \right)$-th element of $\mathbf{Z}_{mk}\in \mathbb{C}^{N\times N}$ being $\left[ \mathbf{Z}_{mk} \right]_{nn^{\prime}}=\mathrm{tr}( \mathbf{\hat{R}}_{mk}^{n^{\prime}n} )$. Moreover, $\mathbf{T}_{kl,\left( 1 \right)}^{\mathrm{L}3}\triangleq\mathrm{diag}( \mathbf{\Gamma }_{kl,1}^{\left( 1 \right)},\cdots ,\mathbf{\Gamma }_{kl,M}^{\left( 1 \right)} ) \in \mathbb{C}^{MN\times MN}$ and the $( m,m^{\prime})$-th submatrix of $\mathbf{T}_{kl,\left( 2 \right)}^{\mathrm{L}3}\in \mathbb{C} ^{MN\times MN}$ is
\begin{equation}
\vspace*{-0.1cm}
\begin{aligned}
\mathbf{T}_{kl,\left( 2 \right)}^{\mathrm{L}3,mm^{\prime}}=\left\{ \begin{array}{c}
	\mathbf{\Gamma }_{kl,m}^{\left( 2 \right)}-\mathbf{\Gamma }_{kl,m}^{\left( 1 \right)},   m=m^{\prime}\\
	\mathbf{\Lambda }_{mkl}\mathbf{\bar{P}}_{l}\mathbf{\Lambda }_{m^{\prime}lk},  m\ne m^{\prime}\\
\end{array} \right.
\end{aligned}
\end{equation}
where the $( n,n^{\prime})$-th element of $\mathbf{\Lambda }_{mkl}\in \mathbb{C}^{N\times N}$, $\mathbf{\Lambda }_{m^{\prime}lk}\in \mathbb{C}^{N\times N}$, $\mathbf{\Gamma }_{kl,m}^{\left( 1 \right)}\in \mathbb{C}^{N\times N}$ and $\mathbf{\Gamma }_{kl,m}^{\left( 2 \right)}\in \mathbb{C}^{N\times N}$ are $\left[ \mathbf{\Lambda }_{mkl} \right] _{nn^{\prime} }=\mathrm{tr}\left( \mathbf{\Theta }_{mkl}^{n^{\prime} n} \right)$, $\left[ \mathbf{\Lambda }_{m^{\prime}lk} \right] _{nn^{\prime}}=\mathrm{tr}\left( \mathbf{\Theta }_{m^{\prime}lk}^{n^{\prime}n} \right)$, $[ \mathbf{\Gamma }_{kl,m}^{( 1 )} ] _{nn^{\prime}}=\sum_{i=1}^N{\sum_{i^{\prime}=1}^N{[ \mathbf{\bar{P}}_{l} ] _{i^{\prime}i}\mathrm{tr}( \mathbf{R}_{ml}^{i^{\prime}i}\mathbf{\hat{R}}_{mk}^{n^{\prime}n} )}}$, and
\begin{equation}\label{eq:Gama2}
\begin{aligned}
\left[ \mathbf{\Gamma }_{kl,m}^{\left( 2 \right)} \right] _{nn^{\prime}}&=\sum_{i=1}^N{\sum_{i^{\prime}=1}^N{\left[ \mathbf{\bar{P}}_l \right] _{i^{\prime}i}\left\{ \mathrm{tr}\left( \mathbf{R}_{ml}^{i^{\prime}i}\mathbf{F}_{mkl,\left( 1 \right)}^{n^{\prime}n} \right) \right.}}\\
&\left. +\tau _{p}^{2}\sum_{q_1=1}^N{\sum_{q_2=1}^N{\left[ \mathrm{tr}\left( \mathbf{\tilde{F}}_{mkl,\left( 2 \right)}^{q_1n}\mathbf{\tilde{R}}_{ml}^{i^{\prime}q_2}\mathbf{\tilde{R}}_{ml}^{q_2i}\mathbf{\tilde{F}}_{mkl,\left( 2 \right)}^{n^{\prime}q_1} \right) +\mathrm{tr}\left( \mathbf{\tilde{F}}_{mkl,\left( 2 \right)}^{q_1n}\mathbf{\tilde{R}}_{ml}^{i^{\prime}q_2} \right) \mathrm{tr}\left( \mathbf{\tilde{F}}_{mkl,\left( 2 \right)}^{n^{\prime}q_2}\mathbf{\tilde{R}}_{ml}^{q_2i} \right) \right]}} \right\},
\end{aligned}
\end{equation}
with $\mathbf{\Theta }_{mkl}=\tau _p\mathbf{R}_{ml}\mathbf{\tilde{P}}_{l}^{H}\mathbf{\Psi }_{mk}^{-1}\mathbf{\tilde{P}}_{k}\mathbf{R}_{mk}$, $\mathbf{\Theta }_{m^{\prime}lk}=\tau _p\mathbf{R}_{m^{\prime}k}\mathbf{\tilde{P}}_{k}^{H}\mathbf{\Psi }_{m^{\prime}k}^{-1}\mathbf{\tilde{P}}_{l}\mathbf{R}_{m^{\prime}l}$, $\mathbf{F}_{mkl,( 1 )}=\tau _p\mathbf{S}_{mk}( \mathbf{\Psi }_{mk}-\tau _p\mathbf{\tilde{P}}_{l}\mathbf{R}_{ml}\mathbf{\tilde{P}}_{l}^{H} ) \mathbf{S}_{mk}^{H}$, $\mathbf{S}_{mk}=\mathbf{R}_{mk}\mathbf{\tilde{P}}_{k}^{H}\mathbf{\Psi }_{mk}^{-1}$, $\mathbf{F}_{mkl,( 2 )}=\mathbf{S}_{mk}\mathbf{\tilde{P}}_{l}\mathbf{R}_{ml}\mathbf{\tilde{P}}_{l}^{H}\mathbf{S}_{mk}^{H}$, $\mathbf{\tilde{R}}_{ml}^{ni}$ and $\mathbf{\tilde{F}}_{mkl,( 2 )}^{ni}$ being $( n,i )$-submatrix of $\mathbf{R}_{ml}^{\frac{1}{2}}$ and $\mathbf{F}_{mkl,( 2 )}^{\frac{1}{2}}$, respectively.
The LSFD coefficient matrix in \eqref{LSFD_op} can also be computed in closed-form as
\setcounter{equation}{29}
\begin{equation}
\mathbf{A}_k=\left( \sum_{l=1}^K\mathbf{T}_{kl,\left( 1 \right)}^{\mathrm{L}3}+\sum_{l\in \mathcal{P}_k}\mathbf{T}_{kl,\left( 2 \right)}^{\mathrm{L}3}+\sigma ^2\mathbf{S}_k \right) ^{-1}\mathbf{Z}_k\mathbf{P}_k.
\end{equation}
\end{theorem}

\begin{proof}
The proof is given in Appendix \ref{appendix_Proof_Th_2}.
\end{proof}

\subsection{Level 2: Local Processing \& Simple Centralized Decoding}
The so-called LSFD method in Level 3 requires a number of the large-scale fading parameters knowledge which may be very large in CF mMIMO. For simplicity, the CPU can alternatively weight the local estimates $\left\{\mathbf{\tilde{x}}_{mk}\!\!:m\!=\!1,\!\cdots\!,\!M \right\}$ by taking the average of them to obtain the final decoding symbol as
$
\mathbf{\hat{x}}_{k}=\sum_{m=1}^M{\frac{1}{M}\mathbf{\tilde{x}}_{mk}}.
$
Note that $\mathbf{\hat{x}}_{k}$ can also be derived from \eqref{eq:S_kn_LSFD} by setting $\mathbf{A}_{mk}=\frac{1}{M}\mathbf{I}_N$ so we can obtain an achievable SE of UE $k$ for Level 2 as the following corollary.
\begin{corollary}
An achievable SE for UE $k$ in Level 2 with MMSE-SIC detectors is
\begin{equation}\label{eq:SE_2}
\mathrm{SE}_{k}^{\left( 2 \right)}=\left( 1-\frac{\tau _p}{\tau _c} \right)\log _2\left| \mathbf{I}_N+\mathbf{D}_{k,\left( 2 \right)}^{H}\mathbf{\Sigma }_{k,\left( 2 \right)}^{-1}\mathbf{D}_{k,\left( 2 \right)} \right|,
\end{equation}
where $\mathbf{D}_{k,\left( 2 \right)}\triangleq \sum\nolimits_{m=1}^M{\mathbb{E}\{ \mathbf{V}_{mk}^{H}\mathbf{H}_{mk}} \} \mathbf{P}_{k}$ and $\mathbf{\Sigma }_{k,\left( 2 \right)}\triangleq\sum_{l=1}^K\sum\nolimits_{m=1}^M{\sum\nolimits_{m^{\prime}=1}^M{\mathbb{E}\{ \mathbf{V}_{mk}^{H}\mathbf{H}_{ml}\mathbf{\bar{P}}_l\mathbf{H}_{m^{\prime}l}^{H}\mathbf{V}_{m^{\prime}k}}} \} -\mathbf{D}_{k,\left( 2 \right)}\mathbf{D}_{k,\left( 2 \right)}^{H}+\sum\nolimits_{m=1}^M\mathbb{E}\{ {\mathbf{V}_{mk}^{H}\mathbf{n}_m\mathbf{n}_{m}^{H}\mathbf{V}_{mk}} \}$.
\end{corollary}
Any combining scheme like MR combining $\mathbf{V}_{mk}=\mathbf{\hat{H}}_{mk}$ or L-MMSE combining as \eqref{eq:L_MMSE_3} is available for \eqref{eq:SE_2}. If MR combining is applied at each AP, we can also derive the closed-form SE expression as the following theorem.
\begin{theorem}\label{Level2_Closed_form}
For MR combining $\mathbf{V}_{mk}=\mathbf{\hat{H}}_{mk}$, we can derive the closed-form SE expression in Level 2 with MMSE-SIC detectors as
\begin{equation}\label{eq:Gama_2_Closed-form}
\mathrm{SE}_{k}^{\left( 2 \right)}=\left( 1-\frac{\tau _p}{\tau _c} \right)\log _2\left| \mathbf{I}_N+\mathbf{D}_{k,\left( 2 \right)}^{H}\mathbf{\Sigma }_{k,\left( 2 \right)}^{-1}\mathbf{D}_{k,\left( 2 \right)} \right|,
\end{equation}
where $\mathbf{D}_{k,\left( 2 \right)}=\sum_{m=1}^M {\mathbf{Z}_{mk}} \mathbf{P}_{k}$ and $\mathbf{\Sigma }_{k,\left( 2 \right)}=\sum_{l=1}^K\mathbf{T}_{kl,\left( 1 \right)}^{\mathrm{L}2}+\sum_{l\in \mathcal{P}_k}\mathbf{T}_{kl,\left( 2 \right)}^{\mathrm{L}2}-\mathbf{D}_{k,\left( 2 \right)}\mathbf{D}_{k,\left( 2 \right)}^{H}+ \sigma ^2\sum_{m=1}^M{\mathbf{Z}_{mk}}$ with $\mathbf{T}_{kl,\left( 1 \right)}^{\mathrm{L}2}=\sum_{m=1}^M{\mathbf{\Gamma }_{kl,m}^{\left( 1 \right)}}$ and
$
\mathbf{T}_{kl,\left( 2 \right)}^{\mathrm{L}2}=\sum_{m=1}^M{\left( \mathbf{\Gamma }_{kl,m}^{\left( 2 \right)}-\mathbf{\Gamma }_{kl,m}^{\left( 1 \right)} \right)}+\sum_{m=1}^M{\sum_{m^{\prime}=1,m\ne m^{\prime}}^M{\mathbf{\Lambda }_{mkl}\mathbf{\bar{P}}_l\mathbf{\Lambda }_{m^{\prime}lk}}}.
$
Definitions of matrices above are the same as that of Theorem~\ref{theorem3}.
\end{theorem}

Note that this processing scheme was widely investigated in existing works on CF mMIMO with multi-antenna UEs, e.g. \cite{113,8901451,8811486}, but relies upon overly idealistic assumption of i.i.d. Rayleigh fading channels. The results derived in this subsection hold for arbitrary channel structures (arbitrary $\mathbf{R}_{mk}$), so can easily reduce to the i.i.d. Rayleigh fading channel by changing $\mathbf{W}_{mk}$ to be \eqref{IID_Rayleigh}. Then, we have $\mathbf{H}_{mk}=\sqrt{\beta _{mk}}\mathbf{H}_{mk,\mathrm{iid}}$, $\mathbf{R}_{mk}=\beta _{mk}\mathbf{I}_{LN}$, $\mathbf{\Psi }_{mk}=\sum\nolimits_{l\in \mathcal{P} _k}^{}{\beta _{ml}\tau _p\mathbf{\tilde{\Omega}}_l}\mathbf{\tilde{\Omega}}_{l}^{H}+\sigma ^2\mathbf{I}_{LN}$ and $\mathbf{\hat{R}}_{mk}=\tau _p\beta _{mk}^{2}\mathbf{\tilde{\Omega}}_{k}^{H}\mathbf{\Psi }_{mk}^{-1}\mathbf{\tilde{\Omega}}_k$, respectively. So we can derive the corresponding SE expressions based on i.i.d. Rayleigh fading channels by plugging above terms in Theorem~\ref{Level2_Closed_form}.

\subsection{Level 1: Small-Cell Network}
As for the last processing scheme, both the channel estimation and signal decoding are implemented locally in one particular AP. The decoding can be done locally by APs using local channel estimates so nothing is exchanged to the CPU. In this case, the CF mMIMO network is truly distributed and turns into a small-cell network. As in \cite{[162]}, the macro-diversity is achieved by selecting the best AP that achieves the highest SE to a specific UE. The achievable SE of UE $k$ at Level 1 is given as follows.

\begin{corollary}
An achievable SE for UE $k$ in Level 1 with MMSE-SIC detectors is given by
\begin{align}\label{eq:SE_1}
\mathrm{SE}_{k}^{\left( 1 \right)}=\underset{m\in \left\{ 1,\cdots ,M \right\}}{\max}\underset{\triangleq \mathrm{SE}_{mk}^{\left( 1 \right)}}{\underbrace{\left( 1-\frac{\tau _p}{\tau _c} \right) \mathbb{E} \left\{ \log _2\left| \mathbf{I}_N+\mathbf{D}_{mk,\left( 1 \right)}^{H}\mathbf{\Sigma }_{mk,\left( 1 \right)}^{-1}\mathbf{D}_{mk,\left( 1 \right)} \right| \right\} }},
\end{align}
where $\mathbf{D}_{mk,\left( 1 \right)}\triangleq\mathbf{V}_{mk}^{H}\mathbf{\hat{H}}_{mk}\mathbf{P}_{k}$ and
$
\mathbf{\Sigma }_{k,\left( 1 \right)}\triangleq\mathbf{V}_{mk}^{H}( \sum_{l=1}^K{(\mathbf{\hat{H}}_{ml}\mathbf{\bar{P}}_l\mathbf{\hat{H}}_{ml}^{H}+\mathbf{C}_{ml}^{\prime})}-\mathbf{\hat{H}}_{mk}\mathbf{\bar{P}}_k\mathbf{\hat{H}}_{mk}^{H}+\sigma ^2\mathbf{I}_L) \mathbf{V}_{mk}.
$
The maximum value of $\mathrm{SE}_{mk}^{\left( 1 \right)}$ is achieved with L-MMSE combining matrix by \eqref{eq:L_MMSE_3} as \begin{equation}\label{eq:SE_mk_max}
\begin{aligned}
\mathrm{SE}_{mk}^{\left( 1 \right)}=\left( 1-\frac{\tau _p}{\tau _c} \right) \mathbb{E}\left\{ \log _2\left| \mathbf{I}_N+\mathbf{P}_{k}^{H}\mathbf{\hat{H}}_{mk}^{H}\left( \sum_{l=1,l\ne k}^K{\mathbf{\hat{H}}_{ml}\mathbf{\bar{P}}_l\mathbf{\hat{H}}_{ml}^{H}}+\sum_{l=1}^K{\mathbf{C}_{ml}^{\prime}}+\sigma ^2\mathbf{I}_L \right) ^{-1}\mathbf{\hat{H}}_{mk}\mathbf{P}_k \right| \right\}.
\end{aligned}
\end{equation}
\end{corollary}
\begin{proof}
We can view this scenario as Corollary~\ref{SE_4} with one particular AP decoding the signal. So we can compute the achievable SE and obtain the maximum value in the same way as Theorem~\ref{Theorem1}.
\end{proof}

\begin{remark}
We notice that any combining scheme like MR combining or MMSE (L-MMSE) combining can be adopted in Level 4 to Level 1 and all the achievable SE expressions can be computed using Monte-Carlo simulations.
As for Level 3 and Level 2, we can derive the closed-form SE expressions with MR combining as \eqref{eq:SE_3_closed_form} and \eqref{eq:Gama_2_Closed-form}. However, as for the L-MMSE combining, we can not derive a closed-form SE expression due to the inverse matrix that contains random matrices.
\end{remark}
\begin{remark}
All results derived in this section are applicable for any Rayleigh fading channel model $\mathbf{h}_{mk}\sim \mathcal{N} _{\mathbb{C}}\left( \mathbf{0},\mathbf{R}_{mk} \right) $ with arbitrary forms of $\mathbf{R}_{mk}$, such as some special cases described in Section~\ref{Special_channel_cases}, so undoubtedly provide many important guidance for the performance analysis of the CF mMIMO system with multi-antenna UEs.
\end{remark}
\begin{remark}
Note that all expressions based on multi-antenna UEs are not simple extensions from the results of single-antenna UEs in \cite{[162]}. The presented analysis is implemented in the matrix domain, and we detect the data vector $\mathbf{\check{x}}_k$ or $\mathbf{\tilde{x}}_{mk}$ instead of detecting the data symbol at each antenna separately as \cite{113,8811486}\footnote{Under the setting in \cite{113,8811486}, for the fully centralized processing, the data symbol transmitted by the $n$-th antenna of UE $k$ is detected separately as $\check{x}_{k,n}=\mathbf{v}_{k,n}^{H}\mathbf{y}$, where $\check{x}_{k,n}$ is the detect of the data transmitted by the $n$-th antenna of UE $k$ and $\mathbf{v}_{k,n}\in \mathbb{C}^{LM}$ is an arbitrary receive combining vector for the $n$-th antenna of UE $k$. Under this setting, we treat each antenna at a UE as a separate ``virtual" UE when it comes to the data transmission and can follow the similar methods in \cite{[162]} to derive the achievable SE expressions.} and we derive the achievable SE expressions based on the mutual information theory in \cite{hjorungnes2011complex}. More importantly, the feasibilities of conclusions in \cite{[162]} based on single-antenna UEs are validated in this paper with the setting of multi-antenna UEs, such as the optimality of MMSE combining for maximizing the $\mathrm{SE}_{k}^{\left( 4 \right)}$ and the optimality of the LSFD scheme for maximizing the $\mathrm{SE}_{k}^{\left( 3 \right)}$.
\end{remark}

\section{Numerical Results}\label{se:Numerical Results}
We consider APs and UEs are uniformly distributed in a $1\times1\,\text{km}^2$ area with a wrap-around scheme \cite{8187178}. The pathloss is computed by the COST 321 Walfish-Ikegami model as
$
\beta _{mk}\left[ \mathrm{dB} \right] =-34.53-38\log _{10}\left( d_{mk}/{1\,\mathrm{m}} \right) +F_{mk},
$
where $d_{mk}$ is the distance between AP $m$ and UE $k$ (taking an $11\,\text{m}$ height difference into account). We model the shadow fading $F_{mk}$ as in \cite{7827017} with $F_{mk}=\sqrt{\delta _f}a_m+\sqrt{1-\delta _f}b_k$, where $a_m\sim\mathcal{N}(0,\delta _{\text{sf}}^{2} )$ and $b_k\sim \mathcal{N}(0,\delta _{\text{sf}}^{2} )$ are independent random variables and $\delta _f$ is the shadow fading parameter. The covariance functions of $a_m$ and $b_k$ are $\mathbb{E}\{ a_ma_{m'} \} =2^{-\frac{d_{mm'}}{d_{\text{dc}}}}$, $\mathbb{E}\{ b_kb_{k'} \} =2^{-\frac{d_{kk'}}{d_{\text{dc}}}}$ where $d_{mm'}$ and $d_{kk'}$ are the geographical distances between AP $m$-AP $m'$ and UE $k$-UE $k'$, respectively, $d_{\text{dc}}$ is the decorrelation distance depending on the environment. Let $\delta _f=0.5$, $d_{\text{dc}}=100\,\text{m}$ and $\delta _{\text{sf}}=8$ \cite{8809413}.

In practice, $\mathbf{U}_{mk,\mathrm{r}}$ and $\mathbf{U}_{mk,\mathrm{t}}$ are estimated through measurements \cite{1576533}. But for numerical simulations in this paper, we generate $\mathbf{U}_{mk,\mathrm{r}}$ and $\mathbf{U}_{mk,\mathrm{t}}$ randomly.  As for the coupling matrix $\mathbf{W}_{mk}$, we have \cite{9148706}
\begin{equation}\label{Coupling_Matrix}
\mathbf{W}_{mk}=\beta _{mk}\left[ \begin{matrix}
	\frac{LN}{2}&		a\mathbf{1}_{1\times \left( N-1 \right)}\\
	a\mathbf{1}_{\left( L-1 \right) \times 1}&		a\mathbf{1}_{\left( L-1 \right) \times \left( N-1 \right)}\\
\end{matrix} \right]
\end{equation}
for the jointly-correlated Rayleigh fading channel with $a\triangleq \frac{LN}{2\left( LN-1 \right)}$. We investigate communication with $20\,\text{MHz}$ bandwidth and $\sigma ^2=-94\,\text{dBm}$ noise power. All the UEs transmit with the power $200\,\text{mW}$ and the power is assumed to be divided equally between the $N$ antennas of each UE, that is $\mathbf{\Omega}_k=\mathbf{P}_k=\frac{1}{N}\mathbf{I}_N$. Besides, each coherence block contains $\tau _c=200$ channel uses and $\tau _p=KN$, unless mentioned. Moreover, an effective pilot assignment approach based on the full correlation matrix $\mathbf{R}_{mk}$ is implemented, where first $\lfloor \tau _p/N \rfloor $ UEs are assigned with each pilot matrix randomly and other UEs are assigned with their pilot matrices that achieve the least interference to UEs in the current pilot set.

\begin{remark}
The coupling matrix in \eqref{Coupling_Matrix} has one entry scaled with $\mathcal{O} \left( LN \right)$, which means that there are dominant transmit-receive eigenpairs capturing half of the channel power. The remaining channel power is assumed to be divided equally to the other entries in the coupling matrix.
\end{remark}

\subsection{Effects of the Number of Antennas Per UE}\label{Number_UE_antennas}
We firstly investigate the effects of the number of antennas per UE. Figure \ref{CDF_four_levels} shows the complementary cumulative distribution function (CCDF) of the per-user SE for four processing schemes investigated in this paper over MMSE (L-MMSE) combining and MR combining with $M=40, K=20$ and $L=4$. For MMSE or L-MMSE combining shown in Fig. \ref{CDF_four_levels} (a), we observe that Level 4 undoubtedly outperforms other schemes since it can suppress the interference by all the collective channel estimates through MMSE combining and MMSE-SIC detectors. For MR combing, compared to Fig. \ref{CDF_four_levels} (a), all processing schemes except Level 1 suffer from a large SE loss, since MR combining can not suppress the interference effectively. Besides, for Level 3 and Level 2, markers ``$\circ$" generated by analytical results in \eqref{eq:SE_3_closed_form} and \eqref{eq:Gama_2_Closed-form} overlap with the curves generated by simulations, respectively, validating the accuracy of our derived closed-form SE expressions. Moreover, the performance gap between Level 2 and Level 3 for L-MMSE combining is smaller than that of MR combining, which indicates the combination of L-MMSE combining and MMSE-SIC detector is effective for Level 2 to achieve the approaching SE performance to Level 3.

\begin{figure}[t]\centering
\vspace{0.3cm}
\subfigure[ MMSE (L-MMSE) combining]{
\begin{minipage}{8cm}\centering
\includegraphics[scale=0.6]{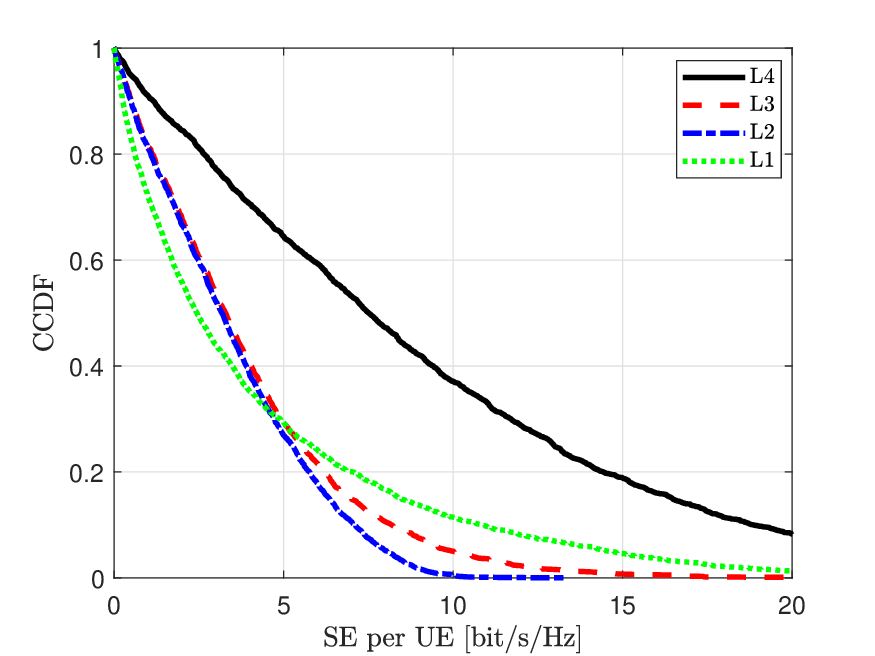}
\end{minipage}}
\subfigure[ MR combining]{
\begin{minipage}{8cm}\centering
\includegraphics[scale=0.6]{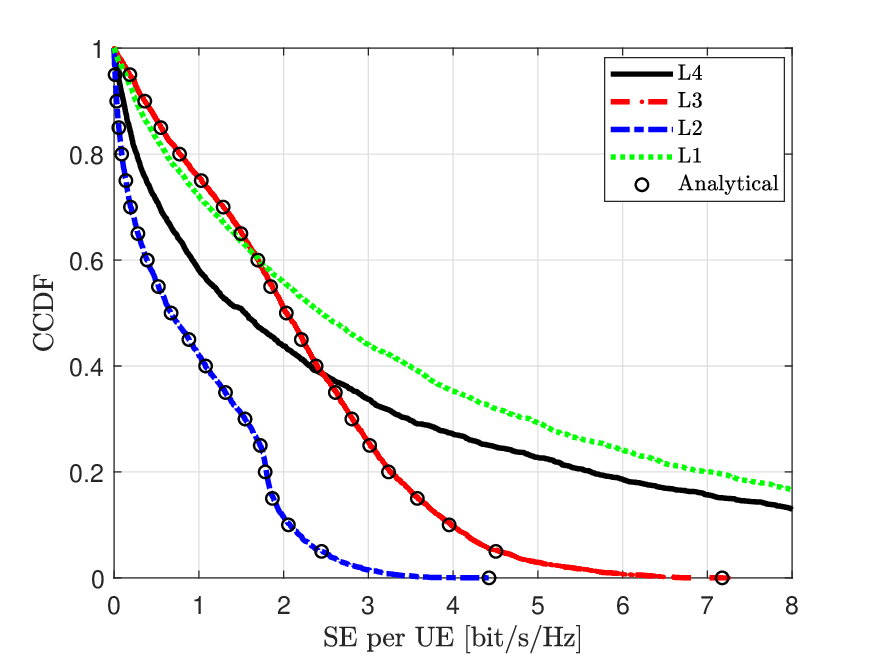}
\end{minipage}}
\caption{CCDF of SE per UE for MMSE (L-MMSE) combining and MR combining with $M=40$, $K=20$ and $L=4$.
\label{CDF_four_levels}}
\end{figure}

Figure \ref{80_likely_SE} shows the $80\%$-likely per-user SE as a function of the number of antennas per UE $N$ for four processing schemes over MMSE (L-MMSE) combining and MR combining with $M=40, K=20$ and $L=4$. From the view of $80\%$ likely SE points, when using MMSE or L-MMSE combining, we observe that Level 4 outperforms other schemes, while Level 1 gives the lowest SEs. Besides, the performance gap between Level 3 and Level 2 using L-MMSE combining is smaller than that of MR combining. When using MR combining, Level 3 achieves the highest SEs, while Level 4 gives a poor SE performance. Moreover, we notice that the SEs reach the maximum values with particular ``$N^*$" noted in the figures, then decrease with the increase of $N$. This phenomenon indicates that the increase of the number of antennas per UE may give rise to the SE degradation.
The reason for this performance degradation is that increasing $N$ will reduce the pre-log factor $\left( \tau _c-\tau _p \right)/ \tau_c$ in SE expressions and the decrease incurred by the pre-log factor outweighs the gain in having more antennas at the UE-side. Besides, without any uplink precoding scheme and power control method, UEs cannot make full use of having additional antennas to achieve better SE performance. So it is vital to investigate the uplink precoding scheme and power control method for CF mMIMO with multi-antenna UEs in the future work.

\begin{figure}[t]\centering
\vspace{0.3cm}
\subfigure[ MMSE (L-MMSE) combining]{
\begin{minipage}{8cm}\centering
\includegraphics[scale=0.6]{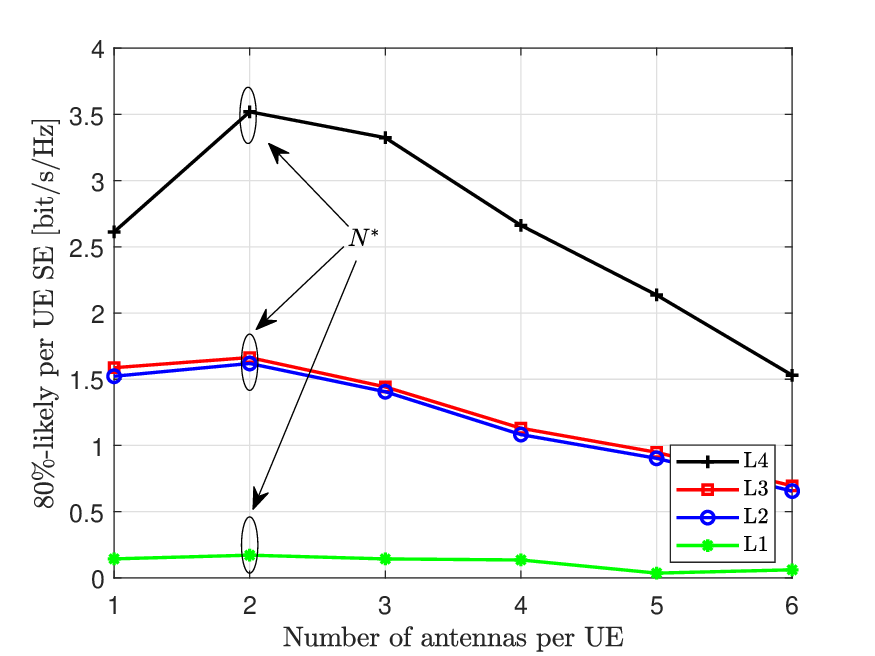}
\end{minipage}}
\subfigure[ MR combining]{
\begin{minipage}{8cm}\centering
\includegraphics[scale=0.6]{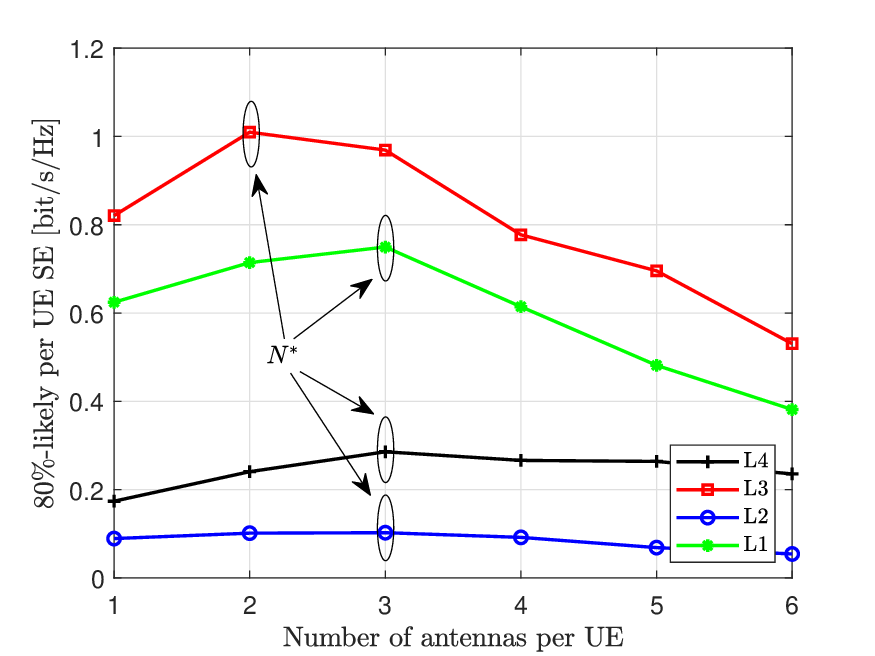}
\end{minipage}}
\caption{$80\%$-likely per-user SE for MMSE (L-MMSE) combining and MR combining as a function of $N$ with $M=40$, $K=20$ and $L=4$.
\label{80_likely_SE}}
\end{figure}

\subsection{Effects of the Length of the Resource Block} \label{Resource_block}
In this subsection, we investigate the effects of the length of the resource block, such as the length of the coherence block $\tau _c$ and the length of duration per coherence block for the pilot transmission $\tau _p$. Figure \ref{Avr_SE_L4_tau_c} considers the average SE as a function of the number of antennas per UE $N$ for Level 4 with MMSE combining over different $\tau _c$. We notice that higher $\tau _c$ undoubtedly achieves higher SE for the increase of the effective transmission ratio. Moreover, we observe that $N^*$ leading to the maximum SE increases as $\tau _c$ increases, which indicates that the average SE benefits from additional UE antennas when the length of the coherence block is long enough, otherwise, additional antennas may lead to the SE degradation.

Figure \ref{Avr_SE_L4_tau_p} shows the average SE as a function of $N$ for Level 4 with MMSE combining over different $\tau _p$. Note that the fact that $\tau _p<KN$ leads to the pilot contamination. For the scenario with $\tau _p=KN/4$, the number of elements in $\mathcal{P}_k$ is $4$, which means that every four UEs will be assigned with a similar pilot matrix. With high $N$ such as $N=6$, the SE given by the scenario with no pilot contamination ($\tau _p=KN$) is lower than that of the scenario with the pilot contamination ($\tau _p=KN/4$ or $\tau _p=KN/10$). Moreover, we can observe a trade-off between the pilot length and the SE performance: when $N$ is small, higher $\tau _p$ such as $\tau _p=KN$ can achieve better SE performance but the SE reduces with the increase of $N$ after reaching the maximum SE value with $N^*$. However, with lower $\tau _p$ such as $\tau _p=KN/10$, the effective transmission ratio $\left( \tau _c-\tau _p \right)/ \tau_c$ reduces slowly with the increase of $N$ so $N^*$ giving the maximum SE value is higher than that of high $\tau _p$. These important insights can also confirm that it is not always worth increasing $N$ since the increase of $N$ may lead to the SE performance degradation.
\begin{figure}[t]
\centering
\includegraphics[scale=0.6]{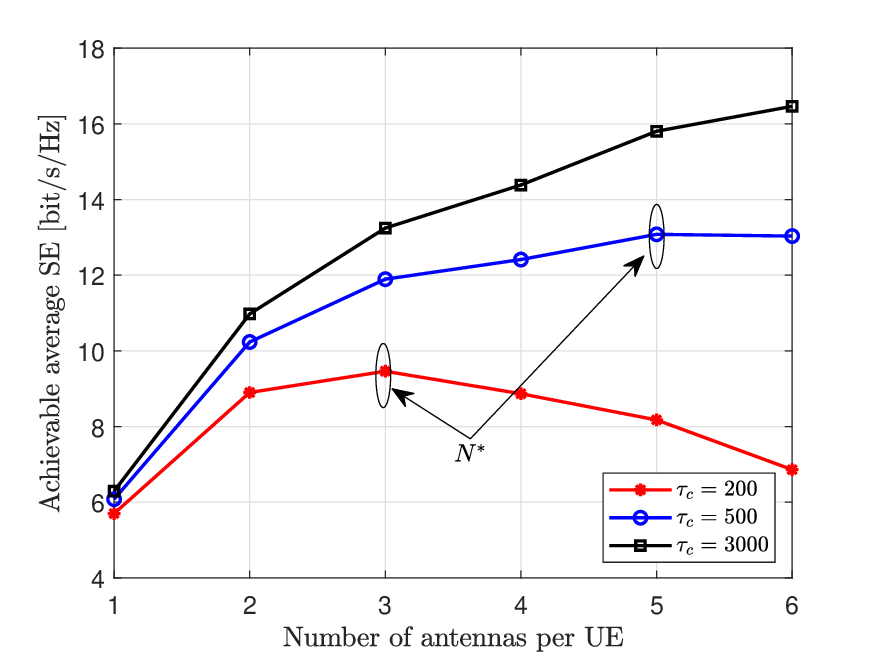}
\caption{Average SE against the number of antennas per UE for Level 4 with MMSE combining over different $\tau_c$ with $M=40$, $K=20$ and $L=4$.
\label{Avr_SE_L4_tau_c}}
\end{figure}
\begin{figure}[t]
\centering
\includegraphics[scale=0.6]{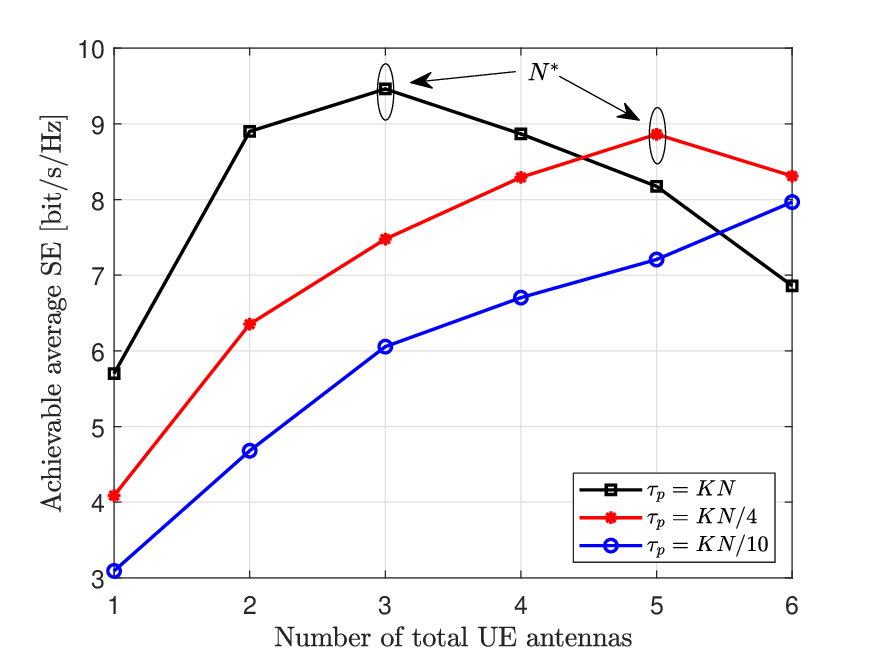}
\caption{Average SE against the number of antennas per UE for Level 4 with MMSE combining over different $\tau_p$ with $M=40$, $K=20$ and $L=4$.
\label{Avr_SE_L4_tau_p}}
\end{figure}

\subsection{Effects of the Number of UEs} \label{Impacts_UE}
Note that the number of the total data streams transmitted equals $KN$. For a particular number of data streams $KN$, $KN$ single-antenna UEs or fewer multi-antenna UEs can be scheduled. Figure \ref{SE_KN} compares the achievable sum SE as a function of the number of total UE antennas $KN$ for Level 2 with MR combining over different system implementations. Figure \ref{SE_KN} shows that for any given $KN$, scheduling $KN$ single-antenna UEs is always beneficial. Besides, the achievable sum SE reaches the maximum value with the optimal $KN$ around $60$.
\begin{figure}[t]
\centering
\includegraphics[scale=0.6]{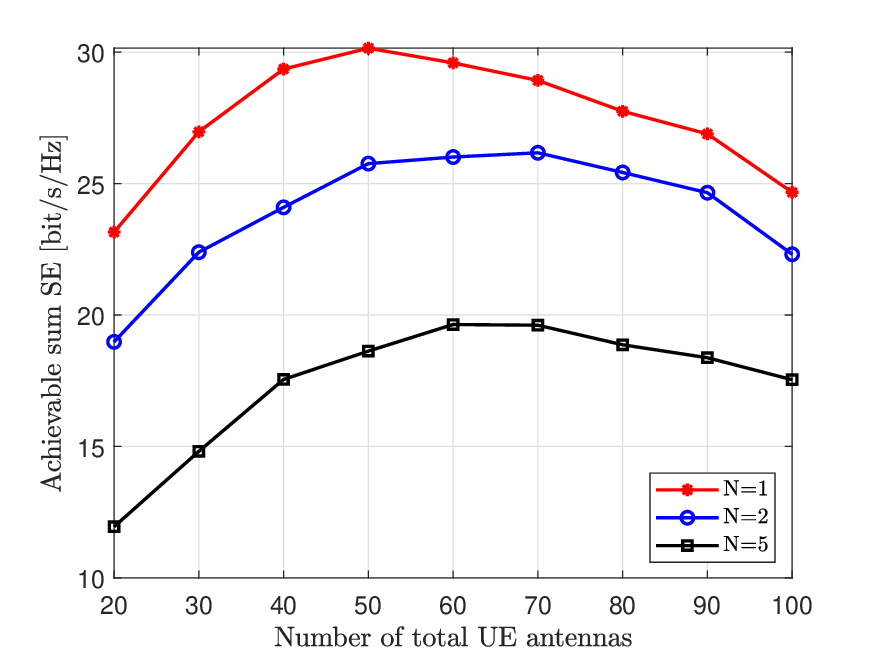}
\caption{Sum SE against the number of total UE antennas for Level 2 with MR combining over $M=40$, $K=20$ and $L=2$.
\label{SE_KN}}
\end{figure}
\begin{figure}[t]
\centering
\includegraphics[scale=0.6]{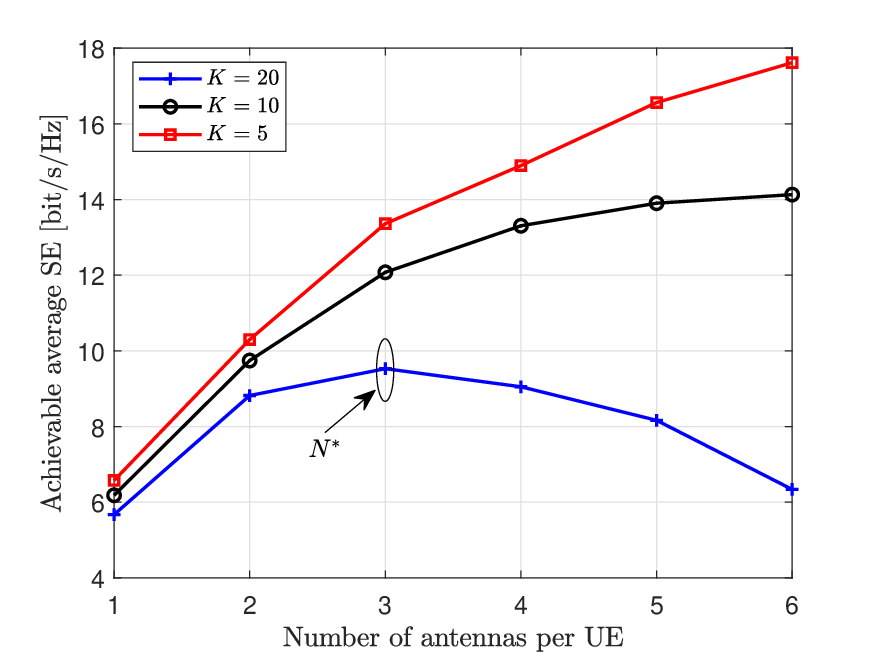}
\caption{Average SE for Level 4 over MMSE combining as a function of $N$ with $M=40$, $K=[5,10,20]$ and $L=4$.
\label{Average_SE_different_K}}
\end{figure}

Figure. \ref{Average_SE_different_K} shows the average SE as a function of the number of antennas per UE $N$ for Level 4 with MMSE combining with $M=40$, $K=[5,10,20]$, and $L=4$. We notice that, with a small or moderate number of UEs compared with the number of APs and total antennas of APs (such as $K=5$ or $K=10$), the average SE performance benefits from additional UE antennas compared with the scenario of a large number of UEs (such as $K=20$). With $K=5$ or $K=10$, compared with $N=1$, equipping with $6$ antennas per UE can achieve about $168\%$ and $129\%$ average SE improvement, respectively. So the SE can greatly benefit from equipping with multiple UE antennas to increase the spatial multiplexing and SE performance significantly in lightly or medium loaded systems with few UEs.


\subsection{Impacts of Channel Models and the Number of Antennas Per AP} \label{Impacts_Channel_Model}
Next, we discuss the impacts of different channel models. The achievable sum SE as a function of the number of APs $M$ for the LSFD scheme with MR combining over different channel models is shown in Fig. \ref{SE_channel_model}.
As observed, the achievable sum SE undoubtedly increases with the increase of $M$. The Kronecker model yields lower SE than the one of the Weichselberger model, for the reason for neglecting the joint spatial correlation feature of the channel. The uncorrelated Rayleigh fading channel achieves higher SE than the one of the Weichselberger model. Besides, we investigate the SE performance over other channel models related to the practical physical radio environments in Section \ref{Practical_Channel_Models}.
Note that $\mathbf{W}_{mk}$ with only a single row achieves the lowest SE performance since only a single eigenmode at the AP is active \cite{1576533}. Furthermore, we notice that markers ``$\times $" generated by analytical results overlap with the curves generated by simulations, respectively, which verifies the accuracy of our derived SE closed-form expressions and, more importantly, implies that our derived expressions hold for any Rayleigh fading channel models.
\begin{figure}[t]
\centering
\includegraphics[scale=0.6]{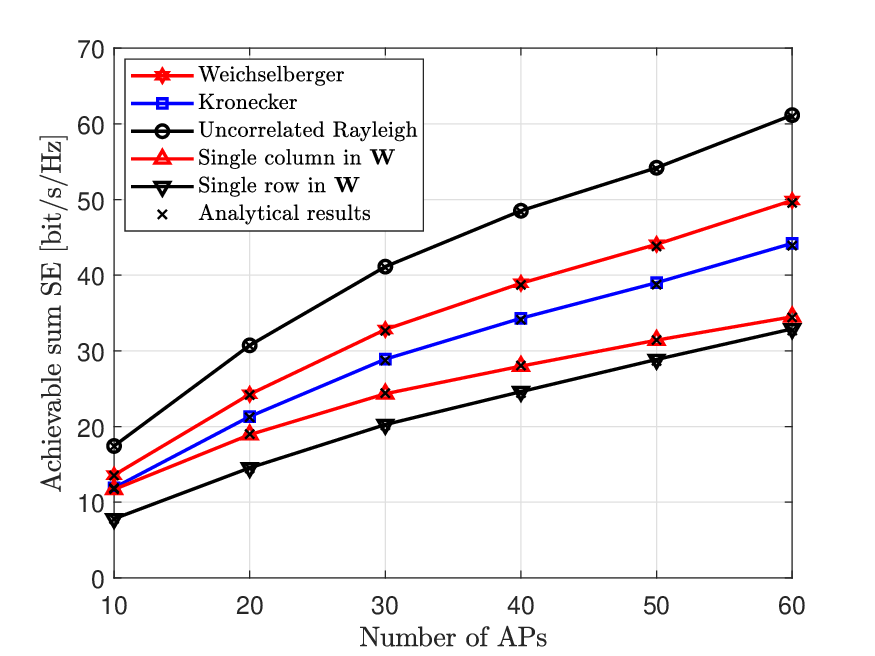}
\caption{Sum SE against the number of APs for Level 3 over MR combining and different channel models with $K=10$, $L=4$, $N=4$ and $\tau _p=KN/2$.
\label{SE_channel_model}}
\end{figure}
\begin{figure}[t]
\centering
\includegraphics[scale=0.6]{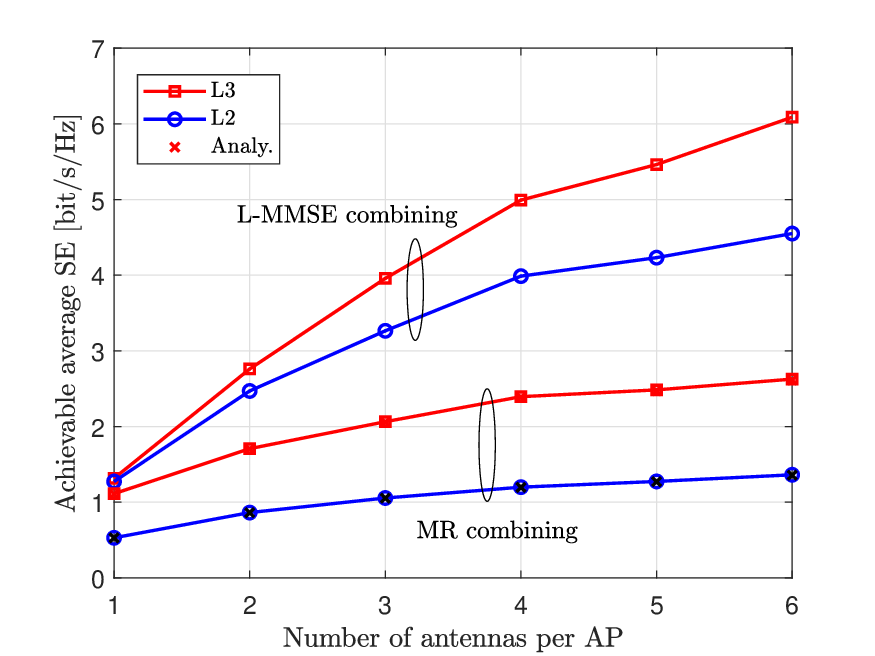}
\caption{Average SE against the number of antennas per AP for Level 3 and Level 2 over L-MMSE combining and MR combining with $M=40, K=20$, $N=2$ and $\tau _p=KN/2$.
\label{SE_L}}
\end{figure}

Figure \ref{SE_L} shows the achievable average SE as a function of the number of antennas per AP for Level 3 and Level 2 over L-MMSE combining and MR combining with $M=40$, $K=20$, $N=2$ and $\tau _p=KN/2$. We notice that the performance gap between L-MMSE combining and MR combining becomes larger with the increase of $L$. Besides, markers ``$\times $" generated by analytical results in \eqref{eq:SE_3_closed_form} and \eqref{eq:Gama_2_Closed-form} overlap with the curves generated by simulations, respectively, validating again the accuracy of our derived closed-form SE expressions with pilot contamination.


\section{Conclusions}\label{se:conclusion}
In this paper, we investigate the uplink SE performance of a CF mMIMO system with both APs and UEs equipped with multiple antennas over the jointly-correlated (or the Weichselberger model) Rayleigh fading channel. We consider four different implementations of CF mMIMO from fully centralized to fully distributed with multi-antenna UEs and derive achievable SE expressions with MMSE-SIC detectors for any combining scheme. Then based on different CSI, we design MMSE combining matrix with full CSI and local MMSE combining matrix with local CSI. Moreover, we prove the optimality for MMSE combining in Level 4, L-MMSE combining in Level 1, and the optimal LSFD coefficients matrix in Level 3 to maximize the respective achievable SE. Besides, with MR combining, we compute the novel closed-form SE expressions for Level 3 and Level 2. In numerical results, we investigate the impact of the number of antennas per UE and compare the SE performance for different processing schemes, combining schemes and channel models. It is greatly important to find that additional UE antennas may degrade the SE performance. And additional UE antennas are beneficial to the SE performance in lightly or medium loaded systems with few UEs. In the future work, we will investigate the uplink precoding scheme with multi-antenna UEs, the power control and allocation scheme for multi-antenna UEs, and scalable CF mMIMO systems with multi-antenna UEs.

\appendix

\subsection{Proof of Corollary 1}\label{appendix_proof_coro_1}
Note that \eqref{eq:Level4_Sk} can be rewritten as
$
\mathbf{\check{x}}_k=\mathbf{V}_{k}^{H}\mathbf{H}_k\mathbf{P}_{k}\mathbf{x}_k+\sum_{l\ne k}^K{\mathbf{V}_{k}^{H}\mathbf{H}_l\mathbf{P}_{l}\mathbf{x}_l}+\mathbf{V}_{k}^{H}\mathbf{n}.
$
So according to the definition of mutual information, we have \cite{1624653}
\begin{equation}\label{eq:mutual_information}
I( \mathbf{x}_k;\mathbf{\check{x}}_k,\mathbf{\hat{H}}_k ) =h(  \mathbf{x}_k |\mathbf{\hat{H}}_k ) -h( \mathbf{x}_k |\mathbf{\check{x}}_k,\mathbf{\hat{H}}_k ),
\end{equation}
where $h\left( \cdot \right) $ denotes the differential entropy. Choosing the potentially suboptimal $\mathbf{x}_k\sim \mathcal{N}_{\mathbb{C}}\left( 0,\mathbf{I}_N \right)$ yields
\begin{equation}\label{eq:mutual_information_1}
h\left( \mathbf{x}_k \right) =\log _2\left| \pi e\mathbf{I}_N \right|.
\end{equation}
Then, the MMSE estimate of $\mathbf{x}_k$ given $\mathbf{\check{x}}_k$ and $\mathbf{\hat{H}}_k$ is
\begin{equation}
\begin{aligned}
\mathbf{\bar{x}}_k&=\mathbb{E}\{  \mathbf{V}_{k}^{H}\mathbf{H}_k\mathbf{P}_{k} |\mathbf{\hat{H}}_k \} \mathbb{E}\{ \mathbf{\check{x}}_k\mathbf{\check{x}}_{k}^{H} |\mathbf{\hat{H}}_k \} ^{-1}\mathbf{\check{x}}_k\\
&=\mathbf{D}_{k,\left( 4 \right)}\mathbf{\bar{\Sigma}}_{k,\left( 4 \right)}^{-1}\mathbf{\check{x}}_k,
\end{aligned}
\end{equation}
where $\mathbf{D}_{k,\left( 4 \right)}=\mathbb{E}\{ \mathbf{V}_{k}^{H}\mathbf{H}_k\mathbf{P}_{k} |\mathbf{\hat{H}}_k \} =\mathbf{V}_{k}^{H}\mathbf{\hat{H}}_k\mathbf{P}_{k}$,
$
\mathbf{\bar{\Sigma}}_{k,\left( 4 \right)}=\mathbb{E}\{ \mathbf{\check{x}}_k\mathbf{\check{x}}_{k}^{H} |\mathbf{\hat{H}}_k \}=\mathbf{V}_{k}^{H}( \sum_{l=1}^K{\mathbf{\hat{H}}_l\mathbf{\bar{P}}_l\mathbf{\hat{H}}_{l}^{H}}+\sum_{l=1}^K{\mathbf{C}_{l}^{\prime}}+\sigma ^2\mathbf{I}_{ML} ) \mathbf{V}_k,
$
and $\mathbf{C}_{l}^{\prime}=\mathbb{E}\{ \mathbf{\tilde{H}}_l\mathbf{\bar{P}}_l\mathbf{\tilde{H}}_{l}^{H} \} =\mathrm{diag}( \mathbf{C}_{1l}^{\prime},\cdots ,\mathbf{C}_{Ml}^{\prime})$
with $\left( j,q \right) $-th element of $\mathbf{C}_{ml}^{\prime}=\mathbb{E}\{ \mathbf{\tilde{H}}_{ml}\mathbf{\bar{P}}_l\mathbf{\tilde{H}}_{ml}^{H} \}$ being
$
[ \mathbf{C}_{ml}^{\prime}] _{jq}=\mathbb{E}\{ \mathbf{\tilde{H}}_{ml}\mathbf{\bar{P}}_l\mathbf{\tilde{H}}_{ml}^{H}\} =\sum_{p_1=1}^N{\sum_{p_2=1}^N{\left[ \mathbf{\bar{P}}_l \right] _{p_2p_1}[ \mathbf{C}_{ml}^{p_2p_1}] _{jq}}}.
$
Moreover, let $\mathbf{\tilde{x}}_k=\mathbf{x}_k-\mathbf{\hat{x}}_k$ denote the estimation error of $\mathbf{x}_k$, then $h(  \mathbf{x}_k |\mathbf{\check{x}}_k,\mathbf{\hat{H}}_k )$ is upper bounded by
\begin{equation}\label{eq:mutual_information_2}
\begin{aligned}
h(  \mathbf{x}_k |\mathbf{\check{x}}_k,\mathbf{\hat{H}}_k) &\leqslant \mathbb{E}\{ \log _2| \pi e\mathbb{E}\{ \mathbf{\tilde{x}}_k\mathbf{\tilde{x}}_{k}^{H} |\mathbf{\hat{H}}_k \} | \}\\
&=\mathbb{E}\{ \log _2| \pi e( \mathbf{I}_N-\mathbf{D}_{k,\left( 4 \right)}\mathbf{\bar{\Sigma}}_{k,\left( 4 \right)}^{-1}\mathbf{D}_{k,\left( 4 \right)}^{H} ) | \}.
\end{aligned}
\end{equation}
Plugging \eqref{eq:mutual_information_1} and \eqref{eq:mutual_information_2} into \eqref{eq:mutual_information} and using the matrix inversion lemma, we have
\begin{equation}
I( \mathbf{x}_k;\mathbf{\check{x}}_k,\mathbf{\hat{H}}_k ) \geqslant \mathbb{E}\{ \log _2 | \mathbf{I}_N+\mathbf{D}_{k,\left( 4 \right)}^{H}\mathbf{\Sigma }_{k,\left( 4 \right)}^{-1}\mathbf{D}_{k,\left( 4 \right)} | \} ,
\end{equation}
where $\mathbf{\Sigma }_{k,\left( 4 \right)}=\mathbf{V}_{k}^{H}( \sum_{l=1}^K{\mathbf{\hat{H}}_l\mathbf{\bar{P}}_l\mathbf{\hat{H}}_{l}^{H}}+\sum_{l=1}^K{\mathbf{C}_{l}^{\prime}}+\sigma ^2\mathbf{I}_{ML}) \mathbf{V}_k-\mathbf{D}_{k,\left( 4 \right)}\mathbf{D}_{k,\left( 4 \right)}^{H}$. So we can derive an achievable SE of UE $k$ as \eqref{eq:SE_4}.

\subsection{Proof of Theorem 1}\label{appendix_Proof_Th_1}
Following from \cite{tse2005fundamentals}, we can rewrite the received signal in \eqref{eq:Level4_y} as
\begin{equation}
\begin{aligned}
\mathbf{y}=\mathbf{\hat{H}}_k\mathbf{P}_k\mathbf{x}_k+\underset{\mathbf{v}}{\underbrace{\mathbf{\tilde{H}}_k\mathbf{P}_k\mathbf{x}_k+\sum_{l\ne k}^K{\mathbf{H}_l\mathbf{P}_l\mathbf{x}_l}+\mathbf{n}}},
\end{aligned}
\end{equation}
where $\mathbf{v}\triangleq \mathbf{\tilde{H}}_k\mathbf{P}_{k}\mathbf{x}_k+\sum_{l\ne k}^K{\mathbf{H}_l\mathbf{P}_{l}\mathbf{x}_l}+\mathbf{n}$ is a complex circular symmetric colored noise with an invertible covariance matrix
\begin{equation}
\begin{aligned}
\mathbf{\Xi }_k&=\mathbb{E}\{  \mathbf{vv}^H |\mathbf{\hat{H}}_k\}\\ &=\sum_{l=1}^K{\mathbf{\hat{H}}_l\mathbf{\bar{P}}_l\mathbf{\hat{H}}_{l}^{H}}-\mathbf{\hat{H}}_k\mathbf{\bar{P}}_k\mathbf{\hat{H}}_{k}^{H}+\sum_{l=1}^K{\mathbf{C}_{l}^{\prime}}+\sigma ^2\mathbf{I}_{ML}.
\end{aligned}
\end{equation}


We firstly whiten the noise as
\setcounter{equation}{42}
\begin{equation}\label{eq:White_Signal}
\mathbf{\Xi }_{k}^{-\frac{1}{2}}\mathbf{y}=\mathbf{\Xi }_{k}^{-\frac{1}{2}}\mathbf{\hat{H}}_k\mathbf{P}_{k}\mathbf{x}_k+\mathbf{\tilde{v}},
\end{equation}
where $\mathbf{\tilde{v}}\triangleq \mathbf{\Xi }_{k}^{-\frac{1}{2}}\mathbf{v}$ becomes white.

Next, we project \eqref{eq:White_Signal} in the direction of $\mathbf{\Xi }_{k}^{-\frac{1}{2}}\mathbf{\hat{H}}_k\mathbf{P}_{k}$ to get an effective scalar channel as
\begin{equation}
\begin{aligned}
( \mathbf{\Xi }_{k}^{-\frac{1}{2}}\mathbf{\hat{H}}_k\mathbf{P}_k) ^H\mathbf{\Xi }_{k}^{-\frac{1}{2}}\mathbf{y}=( \mathbf{\hat{H}}_k\mathbf{P}_k ) ^H\mathbf{\Xi }_{k}^{-1}\cdot ( \mathbf{\hat{H}}_k\mathbf{P}_k\mathbf{x}_k+\mathbf{v}) .
\end{aligned}
\end{equation}

Finally, the optimal receive combining matrix can be represented as
\begin{equation}\label{eq:Optimal_Combining}
\begin{aligned}
\mathbf{V}_k&=\mathbf{\Xi }_{k}^{-1}\mathbf{\hat{H}}_k\mathbf{P}_{k}\\
&=\left(\sum_{l=1}^K{\mathbf{\hat{H}}_l\mathbf{\bar{P}}_l\mathbf{\hat{H}}_{l}^{H}}\!-\!\mathbf{\hat{H}}_k\mathbf{\bar{P}}_k\mathbf{\hat{H}}_{k}^{H}\!+\!\sum_{l=1}^K{\mathbf{C}_{l}^{\prime}}\!+\!\sigma ^2\mathbf{I}_{ML} \right) ^{-1}\!\!\mathbf{\hat{H}}_k\mathbf{P}_{k}.
\end{aligned}
\end{equation}
According to the matrix inversion lemma, we can proof that the combining matrix in \eqref{eq:Optimal_Combining} is equivalent to the MMSE combining matrix in \eqref{eq:MMSE_Com_4} except from having another scaling matrix \cite{8187178}. Note that the SE in \eqref{eq:SE_4} does change if we scale $\mathbf{V}_k$ by any non-zero scalar, so the so-called MMSE combining matrix minimizing the MSE as
\begin{equation}
\begin{aligned}
\mathbf{V}_k=\left( \sum_{l=1}^K{\left( \mathbf{\hat{H}}_l\mathbf{\bar{P}}_l\mathbf{\hat{H}}_{l}^{H}+\mathbf{C}_{l}^{\prime} \right)}+\sigma ^2\mathbf{I}_{ML} \right) ^{-1}\mathbf{\hat{H}}_k\mathbf{P}_k
\end{aligned}
\end{equation}
can also lead to the maximum SE value as
\begin{align}\label{SE_4_Max}
\mathrm{SE}_{k}^{\left( 4 \right)}=\left( 1-\frac{\tau _p}{\tau _c} \right) \mathbb{E}\left\{ \log _2\left| \mathbf{I}_N+\mathbf{P}_k^{H}\mathbf{\hat{H}}_{k}^{H}\left( \sum_{l=1,l\ne k}^K{\mathbf{\hat{H}}_l\mathbf{\bar{P}}_l\mathbf{\hat{H}}_{l}^{H}}+\sum_{l=1}^K{\mathbf{C}_{l}^{\prime}}+\sigma ^2\mathbf{I}_{ML} \right) ^{-1}\mathbf{\hat{H}}_k\mathbf{P}_k \right| \right\}.
\end{align}
\subsection{Proof of the minimization of $\mathrm{MSE}_{k}^{\mathrm{LSFD}}$ by \eqref{LSFD_op}}\label{appendix_minimization_LSFD}
Note that $\mathrm{MSE}_{k}^{\mathrm{LSFD}}=\mathbb{E} \{ \| \mathbf{x}_k-\mathbf{\hat{x}}_k \| ^2| \mathbf{\Theta } \}
=\mathrm{tr}( \mathbb{E} \left\{ ( \mathbf{x}_k-\mathbf{\hat{x}}_k ) ( \mathbf{x}_k-\mathbf{\hat{x}}_k ) ^H| \mathbf{\Theta } \right. \} )$ can be denoted as
\begin{equation}\label{MSE_3}
\begin{aligned}
\mathrm{MSE}_{k}^{\mathrm{LSFD}}&=\mathrm{tr}\left( \mathbb{E} \left\{ \mathbf{x}_k\mathbf{x}_{k}^{H}\left| \mathbf{\Theta } \right. \right\} -\mathbb{E} \left\{ \mathbf{x}_k\mathbf{\hat{x}}_{k}^{H}\left| \mathbf{\Theta } \right. \right\} -\mathbb{E} \left\{ \mathbf{\hat{x}}_k\mathbf{x}_{k}^{H}\left| \mathbf{\Theta } \right. \right\} +\mathbb{E} \left\{ \mathbf{\hat{x}}_k\mathbf{\hat{x}}_{k}^{H}\left| \mathbf{\Theta } \right. \right\} \right)\\
&=\mathrm{tr}\left( \mathbf{I}_N-\mathbf{P}_{k}^{H}\mathbb{E} \left\{ \mathbf{G}_{kk}^{H} \right\} \mathbf{A}_k-\mathbf{A}_{k}^{H}\mathbb{E} \left\{ \mathbf{G}_{kk} \right\} \mathbf{P}_k+\sum_{l=1}^K{\mathbf{A}_{k}^{H}\mathbb{E} \left\{ \mathbf{G}_{kl}\mathbf{\bar{P}}_l\mathbf{G}_{kl} \right\} \mathbf{A}_k} \right).
\end{aligned}
\end{equation}
By computing the partial derivation of $\mathrm{MSE}_{k}^{\mathrm{LSFD}}$ with respect to $\mathbf{A}_{k}^{H}$, we observe that $\mathbf{A}_k$ that minimizes $\mathrm{MSE}_{k}^{\mathrm{LSFD}}$ has the similar form of \eqref{LSFD_op}.

\subsection{Proof of Theorem 2}\label{appendix_Proof_Th_2}
In this part, we compute the closed-form SE expression for Level 3 based on MR combining $\mathbf{V}_{mk}=\mathbf{\hat{H}}_{mk}$ and MMSE-SIC detectors.

We begin with the first term $\mathbf{D}_{k,\left( 3 \right)}=\mathbf{A}_{k}^{H}\mathbb{E}\{ \mathbf{G}_{kk} \} \mathbf{P}_{k}$ with $\mathbb{E}\{ \mathbf{G}_{kk} \} =[ \mathbb{E}\{ \mathbf{V}_{1k}^{H}\mathbf{H}_{1k} \} ;\cdots ;\mathbb{E}\{ \mathbf{V}_{Mk}^{H}\mathbf{H}_{Mk} \} ] =[ \mathbf{Z}_{1k};\cdots ;\mathbf{Z}_{Mk} ]$, where $\mathbf{Z}_{mk}=\mathbb{E}\{ \mathbf{V}_{mk}^{H}\mathbf{H}_{mk} \} =\mathbb{E}\{ \mathbf{\hat{H}}_{mk}^{H}\mathbf{\hat{H}}_{mk} \} \in \mathbb{C}^{N\times N}$ and the $\left( n,n^{\prime} \right) $-th element of $\mathbf{Z}_{mk}$ can be denoted as
$
[ \mathbf{Z}_{mk} ] _{nn^{\prime}}=\mathbb{E}\{ \mathbf{\hat{h}}_{mk,n}^{H}\mathbf{\hat{h}}_{mk,n^{\prime}}\} =\mathrm{tr}( \mathbf{\hat{R}}_{mk}^{n^{\prime}n} ).
$

Then, we compute the second term $\mathbf{S}_k\in \mathbb{C}^{MN\times MN} $ as
\setcounter{equation}{48}
\begin{equation}
\begin{aligned}
\mathbf{S}_k &=\mathrm{diag}( \mathbb{E}\{ \mathbf{V}_{1k}^{H}\mathbf{V}_{1k} \} ,\cdots ,\mathbb{E}\{ \mathbf{V}_{Mk}^{H}\mathbf{V}_{Mk}\} )\\
&=\mathrm{diag}( \mathbf{Z}_{1k},\cdots ,\mathbf{Z}_{Mk} ).
\end{aligned}
\end{equation}

The last term in the denominator $\mathbb{E}\{ \mathbf{G}_{kl}\mathbf{\bar{P}}_l\mathbf{G}_{kl}^{H} \} \in \mathbb{C}^{MN\times MN}$ can be written as \begin{equation}\label{GPG}
\begin{aligned}
\mathbb{E} \left\{ \mathbf{G}_{kl}\mathbf{\bar{P}}_l\mathbf{G}_{kl}^{H} \right\} =\left[ \begin{matrix}
	\mathbb{E} \left\{ \mathbf{V}_{1k}^{H}\mathbf{H}_{1l}\mathbf{\bar{P}}_l\mathbf{H}_{1l}^{H}\mathbf{V}_{1k} \right\}&		\cdots&		\mathbb{E} \left\{ \mathbf{V}_{1k}^{H}\mathbf{H}_{1l}\mathbf{\bar{P}}_l\mathbf{H}_{Ml}^{H}\mathbf{V}_{Mk} \right\}\\
	\vdots&		\ddots&		\vdots\\
	\mathbb{E} \left\{ \mathbf{V}_{Mk}^{H}\mathbf{H}_{Ml}\mathbf{\bar{P}}_l\mathbf{H}_{1l}^{H}\mathbf{V}_{1k} \right\}&		\cdots&		\mathbb{E} \left\{ \mathbf{V}_{Mk}^{H}\mathbf{H}_{Ml}\mathbf{\bar{P}}_l\mathbf{H}_{Ml}^{H}\mathbf{V}_{Mk} \right\}\\
\end{matrix} \right],
\end{aligned}
\end{equation}
and the $( m,m^{\prime})$-submatrix of $\mathbb{E}\{ \mathbf{G}_{kl}\mathbf{\bar{P}}_l\mathbf{G}_{kl}^{H} \}$ is $\mathbb{E}\{ \mathbf{V}_{mk}^{H}\mathbf{H}_{ml}\mathbf{\bar{P}}_l\mathbf{H}_{m^{\prime}l}^{H}\mathbf{V}_{m^{\prime}k} \} $. Then we can compute \\$\mathbb{E}\{ \mathbf{V}_{mk}^{H}\mathbf{H}_{ml}\mathbf{\bar{P}}_l\mathbf{H}_{m^{\prime}l}^{H}\mathbf{V}_{m^{\prime}k} \} $ for all possible AP and UE combinations.

\textbf{Case 1:} $m\ne m^{\prime},l\notin \mathcal{P}_k$

We have $\mathbb{E}\{ \mathbf{V}_{mk}^{H}\mathbf{H}_{ml}\mathbf{\bar{P}}_l\mathbf{H}_{m^{\prime}l}^{H}\mathbf{V}_{m^{\prime}k} \} =0$ since $\mathbf{V}_{mk}$ and $\mathbf{H}_{ml}$ are independent and both have zero mean.

\textbf{Case 2:} $m\ne m^{\prime},l\in \mathcal{P}_k$

We can obtain
$\mathbb{E}\{ \mathbf{V}_{mk}^{H}\mathbf{H}_{ml}\mathbf{\bar{P}}_l\mathbf{H}_{m^{\prime}l}^{H}\mathbf{V}_{m^{\prime}k} \} =\mathbb{E}\{ \mathbf{V}_{mk}^{H}\mathbf{H}_{ml} \} \mathbf{\bar{P}}_l\mathbb{E}\{ \mathbf{H}_{m^{\prime}l}^{H}\mathbf{V}_{m^{\prime}k}\} =\mathbf{\Lambda }_{mkl}\mathbf{\bar{P}}_l\mathbf{\Lambda }_{m^{\prime}lk}$ for the independence of $\mathbb{E}\{ \mathbf{V}_{mk}^{H}\mathbf{H}_{ml} \}$ and $\mathbb{E}\{ \mathbf{H}_{m^{\prime}l}^{H}\mathbf{V}_{m^{\prime}k} \}$. We notice that
\setcounter{equation}{50}
\begin{equation}\label{eq:Lambda_1}
\begin{aligned}
\mathbf{\Lambda }_{mkl}&=\mathbb{E} \{ \mathbf{V}_{mk}^{H}\mathbf{H}_{ml}\} =\mathbb{E}\{ \mathbf{\hat{H}}_{mk}^{H}\mathbf{\hat{H}}_{ml}\}\\
&=\left[ \begin{matrix}
	\mathrm{tr}\left( \mathbf{\Theta }_{mkl}^{11} \right)&		\cdots&		\mathrm{tr}\left( \mathbf{\Theta }_{mkl}^{N1} \right)\\
	\vdots&		\ddots&		\vdots\\
	\mathrm{tr}\left( \mathbf{\Theta }_{mkl}^{1N} \right)&		\cdots&		\mathrm{tr}\left( \mathbf{\Theta }_{mkl}^{NN} \right)\\
\end{matrix} \right] ,
\end{aligned}
\end{equation}
\begin{equation}\label{eq:Lambda_2}
\begin{aligned}
\mathbf{\Lambda }_{m^{\prime}lk}&=\mathbb{E} \{ \mathbf{H}_{m^{\prime}l}^{H}\mathbf{V}_{m^{\prime}k} \} =\mathbb{E} \{ \mathbf{\hat{H}}_{m^{\prime}l}^{H}\mathbf{\hat{H}}_{m^{\prime}k}\}\\
&=\left[ \begin{matrix}
	\mathrm{tr}\left( \mathbf{\Theta }_{m^{\prime}lk}^{11} \right)&		\cdots&		\mathrm{tr}\left( \mathbf{\Theta }_{m^{\prime}lk}^{N1} \right)\\
	\vdots&		\ddots&		\vdots\\
	\mathrm{tr}\left( \mathbf{\Theta }_{m^{\prime}lk}^{1N} \right)&		\cdots&		\mathrm{tr}\left( \mathbf{\Theta }_{m^{\prime}lk}^{NN} \right)\\
\end{matrix} \right]
\end{aligned}
\end{equation}
where $\mathbf{\Theta }_{mkl}\triangleq \mathbb{E} \{ \mathbf{\hat{h}}_{ml}\mathbf{\hat{h}}_{mk}^{H} \} =\tau _p\mathbf{R}_{ml}\mathbf{\tilde{\Omega}}_{l}^{H}\mathbf{\Psi }_{mk}^{-1}\mathbf{\tilde{\Omega}}_{k}\mathbf{R}_{mk}$ and $\mathbf{\Theta }_{m^{\prime}lk}\triangleq \mathbb{E} \{ \mathbf{\hat{h}}_{m^{\prime}k}\mathbf{\hat{h}}_{m^{\prime}l}^{H} \} = \tau _p\mathbf{R}_{m^{\prime}k}\mathbf{\tilde{\Omega}}_{k}^{H}\mathbf{\Psi }_{m^{\prime}k}^{-1}\mathbf{\tilde{\Omega}}_{l}\mathbf{R}_{m^{\prime}l}$.

\textbf{Case 3:} $m=m^{\prime},l\notin \mathcal{P}_k$

In this case, the channel estimate $\mathbf{\hat{H}}_{mk}$ and $\mathbf{H}_{ml}$ are independent and $\mathbf{\Gamma }_{mkl}^{\left( 1 \right)}\triangleq \mathbb{E}\left\{ \mathbf{V}_{mk}^{H}\mathbf{H}_{ml}\mathbf{\bar{P}}_l\mathbf{H}_{ml}^{H}\mathbf{V}_{mk} \right\} \in \mathbb{C}^{N\times N}$ whose $\left( n,n^{\prime}\right)$-th element can be denoted as
$
[ \mathbf{\Gamma }_{mkl}^{\left( 1 \right)} ] _{nn^{\prime}}=\sum_{i=1}^N{\sum_{i^{\prime}=1}^N{[ \mathbf{\bar{P}}_l ] _{i^{\prime}i}\mathbb{E} \{ \mathbf{\hat{h}}_{mk,n}^{H}\mathbf{h}_{ml,i^{\prime}}\mathbf{h}_{ml,i}^{H}\mathbf{\hat{h}}_{mk,n^{\prime}} \}}}
$

Notice that $\mathbf{\hat{h}}_{mk}$ and $\mathbf{h}_{ml}$ are independent, we have
\setcounter{equation}{52}
\begin{equation}
\begin{aligned}
&\mathbb{E} \{ \mathbf{\hat{h}}_{mk,n}^{H}\mathbf{h}_{ml,i^{\prime}}\mathbf{h}_{ml,i}^{H}\mathbf{\hat{h}}_{mk,n^{\prime}} \}\\
&=\mathrm{tr}( \mathbb{E} \{ \mathbf{h}_{ml,i^{\prime}}\mathbf{h}_{ml,i}^{H}\} \mathbb{E} \{ \mathbf{\hat{h}}_{mk,n^{\prime}}\mathbf{\hat{h}}_{mk,n}^{H} \} ) =\mathrm{tr}( \mathbf{R}_{ml}^{i^{\prime}i}\mathbf{\hat{R}}_{mk}^{n^{\prime}n}).
\end{aligned}
\end{equation}
So, we can derive
\begin{equation}
\left[ \mathbf{\Gamma }_{mkl}^{\left( 1 \right)} \right] _{nn^{\prime}}=\sum_{i=1}^N{\sum_{i^{\prime}=1}^N{\left[ \mathbf{\bar{P}}_l \right] _{i^{\prime}i}\mathrm{tr}\left( \mathbf{R}_{ml}^{i^{\prime}i}\mathbf{\hat{R}}_{mk}^{n^{\prime}n} \right)}}.
\end{equation}

\textbf{Case 4:} $m=m^{\prime},l\in \mathcal{P}_k$

In this case, the channel estimate $\mathbf{\hat{H}}_{mk}$ and $\mathbf{H}_{ml}$ are no longer independent. We define $\mathbf{\Gamma }_{mkl}^{\left( 2 \right)}\triangleq \mathbb{E}\{ \mathbf{V}_{mk}^{H}\mathbf{H}_{ml}\mathbf{\bar{P}}_l\mathbf{H}_{ml}^{H}\mathbf{V}_{mk} \} \in \mathbb{C}^{N\times N}$ whose $( n,n^{\prime})$-th element is
\begin{equation}\label{Gama2}
\left[ \mathbf{\Gamma }_{mkl}^{\left( 2 \right)} \right] _{nn^{\prime}}=\sum_{i=1}^N{\sum_{i^{\prime}=1}^N{\left[ \mathbf{\bar{P}}_l \right] _{i^{\prime}i}\mathbb{E} \left\{ \mathbf{\hat{h}}_{mk,n}^{H}\mathbf{h}_{ml,i^{\prime}}\mathbf{h}_{ml,i}^{H}\mathbf{\hat{h}}_{mk,n^{\prime}} \right\}}}.
\end{equation}

Note that $\mathbf{\hat{h}}_{mk}$ and $\mathbf{h}_{ml}$ are no longer independent. Following the similar step in \cite{8620255} and \cite{9276421}, let $\mathbf{x}_{mk}^{\mathrm{p}}\triangleq \mathrm{vec}\left( \mathbf{Y}_{mk}^{\mathrm{p}} \right) -\tau_p\mathbf{\tilde{\Omega}}_{l}\mathrm{vec}\left( \mathbf{H}_{ml} \right) =\mathbf{y}_{mk}^{\mathrm{p}}-\tau _p\mathbf{\tilde{\Omega}}_{l}\mathbf{h}_{ml}$ and $\mathbf{S}_{mk}=\mathbf{R}_{mk}\mathbf{\tilde{\Omega}}_{k}^{H} \mathbf{\Psi }_{mk}^{-1}$, so we can decompose $\mathbb{E}\{ | \mathbf{\hat{h}}_{mk}^{H}\mathbf{h}_{ml} |^2 \}$ as
\begin{equation}\label{Decompose}
\begin{aligned}
\mathbb{E} \left\{ \left| \mathbf{\hat{h}}_{mk}^{H}\mathbf{h}_{ml} \right|^2 \right\} &=\mathbb{E} \left\{ \left| \left[ \mathbf{S}_{mk}\left( \mathbf{x}_{mk}^{\mathrm{p}}+\tau _p\mathbf{\tilde{\Omega}}_l\mathbf{h}_{ml} \right) \right] ^H\mathbf{h}_{ml} \right|^2 \right\}\\
&=\mathbb{E} \left\{ \left| \left( \mathbf{S}_{mk}\mathbf{x}_{mk}^{\mathrm{p}} \right) ^H\mathbf{h}_{ml} \right|^2 \right\} +\tau _{p}^{2}\mathbb{E} \left\{ \left| \left( \mathbf{S}_{mk}\mathbf{\tilde{\Omega}}_l\mathbf{h}_{ml} \right) ^H\mathbf{h}_{ml} \right|^2 \right\} .
\end{aligned}
\end{equation}

We notice that $\mathbf{S}_{mk}\mathbf{x}_{mk}^{\mathrm{p}}\sim \mathcal{N}_{\mathbb{C}}\left( \mathbf{0},\mathbf{F}_{mkl,\left( 1 \right)} \right)$ and $\mathbf{S}_{mk}\mathbf{\tilde{\Omega}}_l\mathbf{h}_{ml}\sim \mathcal{N}_{\mathbb{C}}\left( \mathbf{0},\mathbf{F}_{mkl,\left( 2 \right)} \right)$ where $\mathbf{F}_{mkl,\left( 1 \right)}$ and $\mathbf{F}_{mkl,\left( 2 \right)}$ are given in
\begin{equation}\label{eq:Gama_2_all}
\begin{aligned}
\mathbb{E}\{ \mathbf{\hat{h}}_{mk,n}^{H}\mathbf{h}_{ml,i^{\prime}}\mathbf{h}_{ml,i}^{H}\mathbf{\hat{h}}_{mk,n^{\prime}} \}=\mathbb{E}\{ \left[ \mathbf{S}_{mk}\mathbf{x}_{mk}^{\mathrm{p}} \right] _{n}^{H}\mathbf{h}_{ml,i^{\prime}}\mathbf{h}_{ml,i}^{H}\left[ \mathbf{S}_{mk}\mathbf{x}_{mk}^{\mathrm{p}} \right] _{n^{\prime}} \}+\tau _{p}^{2}\mathbb{E}\{ [ \mathbf{S}_{mk}\mathbf{\tilde{\Omega}}_{l}\mathbf{h}_{ml}] _{n}^{H}\mathbf{h}_{ml,i^{\prime}}\mathbf{h}_{ml,i}^{H}[ \mathbf{S}_{mk}\mathbf{\tilde{\Omega}}_{l}\mathbf{h}_{ml}] _{n^{\prime}} \}.
\end{aligned}
\end{equation}
Then, we can rewrite $\mathbb{E} \{ \mathbf{\hat{h}}_{mk,n}^{H}\mathbf{h}_{ml,i^{\prime}}\mathbf{h}_{ml,i}^{H}\mathbf{\hat{h}}_{mk,n^{\prime}} \}$ as \begin{equation}\label{FF}
\begin{aligned}
\begin{cases}
	\mathbf{F}_{mkl,\left( 1 \right)}=\mathbb{E} \left\{ \mathbf{S}_{mk}\mathbf{x}_{mk}^{\mathrm{p}}\left( \mathbf{x}_{mk}^{\mathrm{p}} \right) ^H\mathbf{S}_{mk}^{H} \right\} =\tau _p\mathbf{S}_{mk}\left( \mathbf{\Psi }_{mk}-\tau _p\mathbf{\tilde{\Omega}}_{l}\mathbf{R}_{ml}\mathbf{\tilde{\Omega}}_{l}^{H} \right) \mathbf{S}_{mk}^{H},\\
	\mathbf{F}_{mkl,\left( 2 \right)}=\mathbb{E} \left\{ \mathbf{S}_{mk}\mathbf{\tilde{\Omega}}_{l}\mathbf{h}_{ml}\mathbf{h}_{ml}^{H}\mathbf{\tilde{\Omega}}_{l}^{H}\mathbf{S}_{mk}^{H} \right\} =\mathbf{S}_{mk}\mathbf{\tilde{\Omega}}_{l}\mathbf{R}_{ml}\mathbf{\tilde{\Omega}}_{l}^{H}\mathbf{S}_{mk}^{H},\\
\end{cases}
\end{aligned}
\end{equation}
where $\left[ \mathbf{v} \right] _n$ is $\left\{ \left( n-1 \right) L+1\sim nL:n=1,\cdots ,N \right\} $-th rows of $\mathbf{v}\in \mathbb{C}^{LN}$.

For the first term, $ \mathbf{S}_{mk}\mathbf{x}_{mk}^{\mathrm{p}}$ and $\mathbf{h}_{ml}$ are independent, we have
\setcounter{equation}{58}
\begin{equation}
\begin{aligned}
&\mathbb{E} \left\{ \left[ \mathbf{S}_{mk}\mathbf{x}_{mk}^{\mathrm{p}} \right] _{n}^{H}\mathbf{h}_{ml,i^{\prime}}\mathbf{h}_{ml,i}^{H}\left[ \mathbf{S}_{mk}\mathbf{x}_{mk}^{\mathrm{p}} \right] _{n^{\prime}} \right\}\\
&=\mathrm{tr}\left( \mathbb{E} \left\{ \mathbf{h}_{ml,i^{\prime}}\mathbf{h}_{ml,i}^{H} \right\} \mathbb{E} \left\{ \left[ \mathbf{S}_{mk}\mathbf{x}_{mk}^{\mathrm{p}} \right] _{n^{\prime}}\left[ \mathbf{S}_{mk}\mathbf{x}_{mk}^{\mathrm{p}} \right] _{n}^{H} \right\} \right)\\ &=\mathrm{tr}\left( \mathbf{R}_{ml}^{i^{\prime}i}\mathbf{F}_{mkl,\left( 1 \right)}^{n^{\prime}n} \right) .
\end{aligned}
\end{equation}

As for the second term, we can express $[ \mathbf{S}_{mk}\mathbf{\tilde{\Omega}}_l\mathbf{h}_{ml} ] _{\{ n,n^{\prime} \}}$ and $\mathbf{h}_{ml,\{ i,i^{\prime} \}}$ into the equivalent form \cite{9148706} as
$[ \mathbf{S}_{mk}\mathbf{\tilde{\Omega}}_l\mathbf{h}_{ml} ] _n=\sum_{q=1}^N{\mathbf{\tilde{F}}_{mkl,( 2 )}^{nq}\mathbf{x}_q}$, $[ \mathbf{S}_{mk}\mathbf{\tilde{\Omega}}_l\mathbf{h}_{ml}] _{n^{\prime}}=\sum_{q=1}^N{\mathbf{\tilde{F}}_{mkl,( 2 )}^{n^{\prime}q}\mathbf{x}_q}$, $\mathbf{h}_{ml,i}=\sum_{q=1}^N{\mathbf{\tilde{R}}_{ml}^{iq}\mathbf{x}_q}$ and $\mathbf{h}_{ml,i^{\prime}}=\sum_{q=1}^N{\mathbf{\tilde{R}}_{ml}^{i^{\prime}q}\mathbf{x}_q}$, respectively, where $\mathbf{\tilde{F}}_{mkl,( 2 )}^{nq}$ denotes the $( n,q )$-submatrix of $( \mathbf{F}_{mkl,( 2 )} ) ^{\frac{1}{2}}$, $\mathbf{\tilde{R}}_{ml}^{iq}$ denotes the $\left( i,q \right)$-submatrix of $\mathbf{R}_{ml}^{\frac{1}{2}}$ and $\mathbf{x}_q\sim \mathcal{N}_{\mathbb{C}}( \mathbf{0},\mathbf{I}_L )$, respectively. To make it clearer, we can rewrite $\mathbf{h}_{ml}$ as
\setcounter{equation}{59}
\begin{equation}
\begin{aligned}
\mathbf{h}_{ml}=\left[ \begin{array}{c}
	\mathbf{h}_{ml,1}\\
	\vdots\\
	\mathbf{h}_{ml,N}\\
\end{array} \right] =\mathbf{R}_{ml}^{\frac{1}{2}}\left[ \begin{array}{c}
	\mathbf{x}_1\\
	\vdots\\
	\mathbf{x}_N\\
\end{array} \right]=\left[ \begin{matrix}
	\mathbf{\tilde{R}}_{ml}^{11}&		\cdots&		\mathbf{\tilde{R}}_{ml}^{1N}\\
	\vdots&		\ddots&		\vdots\\
	\mathbf{\tilde{R}}_{ml}^{N1}&		\cdots&		\mathbf{\tilde{R}}_{ml}^{NN}\\
\end{matrix} \right] \left[ \begin{array}{c}
	\mathbf{x}_1\\
	\vdots\\
	\mathbf{x}_N\\
\end{array} \right] ,
\end{aligned}
\end{equation}
and we also have
\begin{equation}
\begin{aligned}
\mathbf{R}_{ml}^{ij}=\sum_{q=1}^N{\mathbf{\tilde{R}}_{ml}^{iq}( \mathbf{\tilde{R}}_{ml}^{jq} ) ^H}=\sum_{q=1}^N{\mathbf{\tilde{R}}_{ml}^{iq}\mathbf{\tilde{R}}_{ml}^{qj}}.
\end{aligned}
\end{equation}
So we can formulate the second term as
\begin{equation}\label{eq:Term_2_all}
\begin{aligned}
&\mathbb{E}\left\{ \left[ \mathbf{S}_{mk}\mathbf{\tilde{\Omega}}_{l}^{\frac{1}{2}}\mathbf{h}_{ml} \right] _{n}^{H}\mathbf{h}_{ml,i}\mathbf{h}_{ml,i}^{H}\left[ \mathbf{S}_{mk}\mathbf{\tilde{\Omega}}_{l}^{\frac{1}{2}}\mathbf{h}_{ml} \right] _{n^{\prime}} \right\}\\
&=\mathbb{E}\left\{ \left( \sum_{q_1=1}^N{\mathbf{\tilde{F}}_{mkl,\left( 2 \right)}^{nq_1}\mathbf{x}_{q_1}} \right) ^H\left( \sum_{q_2=1}^N{\mathbf{\tilde{R}}_{ml}^{i^{\prime}q_2}\mathbf{x}_{q_2}} \right) \left( \sum_{q_3=1}^N{\mathbf{\tilde{R}}_{ml}^{iq_3}\mathbf{x}_{q_3}} \right) ^H\left( \sum_{q_4=1}^N{\mathbf{\tilde{F}}_{mkl,\left( 2 \right)}^{n^{\prime}q_4}\mathbf{x}_{q_4}} \right) \right\}.
\end{aligned}
\end{equation}

According to \cite{9148706}, if at least one of $q_j,j=1,\cdots 4$ is different from the others, \eqref{eq:Term_2_all} will be $0$ for the circular
symmetry property and the zero mean value of $\mathbf{x}_{q_j}$ and \eqref{eq:Term_2_all} will be non-zero for the case of $q_1=q_2,q_3=q_4$ and $q_1=q_4,q_2=q_3$. For ``$q_1=q_2,q_3=q_4$", \eqref{eq:Term_2_all} can be rewritten as
\begin{equation}\label{eq:Term_2_Case1}
\begin{aligned}
\mathbb{E}\left\{ \left( \sum_{q_1=1}^N{\mathbf{x}_{q_1}^{H}\mathbf{\tilde{F}}_{mkl,\left( 2 \right)}^{q_1n}\mathbf{\tilde{R}}_{ml}^{i^{\prime}q_1}\mathbf{x}_{q_1}} \right) \left( \sum_{q_3=1}^N{\mathbf{x}_{q_3}^{H}\mathbf{\tilde{F}}_{mkl,\left( 2 \right)}^{q_3n^{\prime}}\mathbf{\tilde{R}}_{ml}^{iq_3}\mathbf{x}_{q_3}} \right) ^H \right\}=\sum_{q_1=1}^N{\sum_{q_3=1}^N{\mathrm{tr}\left( \mathbf{\tilde{F}}_{mkl,\left( 2 \right)}^{q_1n}\mathbf{\tilde{R}}_{ml}^{i^{\prime}q_1} \right) \mathrm{tr}\left( \mathbf{\tilde{F}}_{mkl,\left( 2 \right)}^{n^{\prime}q_3}\mathbf{\tilde{R}}_{ml}^{q_3i} \right)}}.
\end{aligned}
\end{equation}
For ``$q_1=q_4,q_2=q_3$", \eqref{eq:Term_2_all} can be rewritten as
\begin{equation}\label{eq:Term_2_Case2}
\begin{aligned}
\mathbb{E} \left\{ \left( \sum_{q_1=1}^N{\mathbf{x}_{q_1}^{H}\mathbf{\tilde{F}}_{mkl,\left( 2 \right)}^{q_1n}} \right) \left( \sum_{q_2=1}^N{\mathbf{\tilde{R}}_{ml}^{i^{\prime}q_2}\mathbf{x}_{q_2}\mathbf{x}_{q_2}^{H}\mathbf{\tilde{R}}_{mkl,\left( 2 \right)}^{q_2i}} \right) \left( \sum_{q_1=1}^N{\mathbf{\tilde{F}}_{mkl,\left( 2 \right)}^{n^{\prime}q_1}\mathbf{x}_{q_1}} \right) \right\}=\sum_{q_1=1}^N{\sum_{q_2=1}^N{\mathrm{tr}\left( \mathbf{\tilde{F}}_{mkl,\left( 2 \right)}^{q_1n}\mathbf{\tilde{R}}_{ml}^{i^{\prime}q_2}\mathbf{\tilde{R}}_{ml}^{q_2i}\mathbf{\tilde{F}}_{mkl,\left( 2 \right)}^{n^{\prime}q_1} \right)}}.
\end{aligned}
\end{equation}

In summary, plugging above results into \eqref{Gama2}, \eqref{Gama2} can be denoted by
\begin{equation}\label{eq:Gama2_Final}
\begin{aligned}
\left[ \mathbf{\Gamma }_{mkl}^{\left( 2 \right)} \right] _{nn^{\prime}}&=\sum_{i=1}^N{\sum_{i^{\prime}=1}^N{\left[ \mathbf{\bar{P}}_l \right] _{i^{\prime}i}\left\{ \mathrm{tr}\left( \mathbf{R}_{ml}^{i^{\prime}i}\mathbf{F}_{mkl,\left( 1 \right)}^{n^{\prime}n} \right) \right.}}\\
&\left. +\tau _{p}^{2}\sum_{q_1=1}^N{\sum_{q_2=1}^N{\left[ \mathrm{tr}\left( \mathbf{\tilde{F}}_{mkl,\left( 2 \right)}^{q_1n}\mathbf{\tilde{R}}_{ml}^{i^{\prime}q_2}\mathbf{\tilde{R}}_{ml}^{q_2i}\mathbf{\tilde{F}}_{mkl,\left( 2 \right)}^{n^{\prime}q_1} \right) +\mathrm{tr}\left( \mathbf{\tilde{F}}_{mkl,\left( 2 \right)}^{q_1n}\mathbf{\tilde{R}}_{ml}^{i^{\prime}q_1} \right) \mathrm{tr}\left( \mathbf{\tilde{F}}_{mkl,\left( 2 \right)}^{n^{\prime}q_2}\mathbf{\tilde{R}}_{ml}^{q_2i} \right) \right]}} \right\}.
\end{aligned}
\end{equation}

Finally, combining all the cases, we can obtain
\setcounter{equation}{65}
\begin{equation}
\begin{aligned}
\mathbb{E}\left\{ \mathbf{G}_{kl}\mathbf{\bar{P}}_l\mathbf{G}_{kl}^{H} \right\} =\mathbf{T}_{kl,\left( 1 \right)}^{\mathrm{L}3}+\begin{cases}
	\mathbf{0}, \ \ \ \quad  l\notin \mathcal{P}_k\\
	\mathbf{T}_{kl,\left( 2 \right)}^{\mathrm{L}3}, l\in \mathcal{P}_k\\
\end{cases}
\end{aligned}
\end{equation}
where $\mathbf{T}_{kl,\left( 1 \right)}^{\mathrm{L}3}=\mathrm{diag}( \mathbf{\Gamma }_{kl,1}^{\left( 1 \right)},\cdots ,\mathbf{\Gamma }_{kl,M}^{\left( 1 \right)} )$ and the $( m,m^{\prime})$-th submatrix of $\mathbf{T}_{kl,\left( 2 \right)}^{\mathrm{L}3}\in \mathbb{C} ^{MN\times MN}$ is
\begin{equation}
\begin{aligned}
\mathbf{T}_{kl,\left( 2 \right)}^{\mathrm{L}3,mm^{\prime}}=\left\{ \begin{array}{c}
	\mathbf{\Gamma }_{kl,m}^{\left( 2 \right)}-\mathbf{\Gamma }_{kl,m}^{\left( 1 \right)},   m=m^{\prime}\\
	\mathbf{\Lambda }_{mkl}\mathbf{\bar{P}}_{l}\mathbf{\Lambda }_{m^{\prime}lk}.  m\ne m^{\prime}\\
\end{array} \right.
\end{aligned}
\end{equation}

\bibliographystyle{IEEEtran}
\bibliography{IEEEabrv,Ref}

\end{document}